\documentclass[11pt,letterpaper]{article}
\usepackage[margin=1in]{geometry}

\usepackage[T1]{fontenc}
\usepackage[utf8]{inputenc}
\usepackage{amsmath,amsthm,amssymb,amsfonts}
\usepackage{mathtools}
\usepackage{microtype}
\usepackage{url}
\usepackage[dvipsnames]{xcolor}
\usepackage{hyperref}

\hypersetup{colorlinks=true,allcolors=blue}
\allowdisplaybreaks[4]

\theoremstyle{plain}
\newtheorem{theorem}{Theorem}[section]
\newtheorem{lemma}[theorem]{Lemma}

\newtheorem{corollary}[theorem]{Corollary}

\newtheorem{proposition}[theorem]{Proposition}

\theoremstyle{definition}
\newtheorem{definition}[theorem]{Definition}

\theoremstyle{remark}

\newtheorem*{remark*}{Remark}

\newcommand{\N}{N} 
\newcommand{\C}{C} 
\newcommand{\bfU}{\mathbf{U}} 
\newcommand{\calA}{\mathcal{A}} 
\newcommand{\calB}{\mathcal{B}} 
\newcommand{\calE}{\mathcal{E}} 
\newcommand{\calJ}{\mathcal{J}} 
\newcommand{\calK}{\mathcal{K}} 
\newcommand{\calL}{\mathcal{L}} 
\newcommand{\calM}{\mathcal{M}} 
\newcommand{\calF}{\mathcal{F}} 
\newcommand{\calH}{\mathcal{H}} 
\newcommand{\calR}{\mathcal{R}} 
\newcommand{\calU}{\mathcal{U}} 

\newcommand{\calW}{\mathcal{W}} 
\newcommand{\sfB}{\mathsf{B}} 
\newcommand{\sfS}{\mathsf{S}} 
\newcommand{\frakS}{\mathfrak{S}} 
\newcommand{\RR}{\mathbb{R}}
\newcommand{\ZZ}{\mathbb{Z}}
\newcommand{\FF}{\mathbb{F}}

\newcommand{\rd}{\mathrm{rd}} 
\newcommand{\GreedyAV}{\mathrm{GreedyAV}} 
\newcommand{\PAV}{\mathrm{PAV}}
\newcommand{\MES}{\mathrm{MES}}

\DeclareMathOperator{\supp}{supp} 
\DeclareMathOperator{\fund}{fund} 

\DeclareMathOperator{\E}{\mathbb{E}}
\DeclareMathOperator{\e}{\mathrm{e}}

\begin{document}

\title{Multi-Winner Elections: Justified Representation, strategyproofness, and Risk-Avoiding Truthfulness}

\author{%
  Yizhou Ai \\
  University of Toronto \\
  \texttt{yizhou.ai@mail.utoronto.ca} \\
  \and
  Biaoshuai Tao \\
  Shanghai Jiao Tong University \\
  \texttt{bstao@sjtu.edu.cn} \\
}
\date{}

\maketitle

\begin{abstract}
We study approval-based multi-winner elections with \emph{Justified Representation} (JR) when voters strategically report their ballots.
We prove that there does not exist a strategy-proof mechanism that outputs JR committees, even when the mechanism can be randomized and only \emph{Ex-Ante strategyproofness} is required.
The impossibility already holds with three voters, four candidates, and target committee size three, even when the mechanism may select fewer than three candidates.
It applies to every common strictly increasing cardinality-based utility function and scales to every target size $k\geq3$, every number of voters divisible by $k$, and every candidate-set size at least $k+1$.

Motivated by our negative results, we then ask for weaker strategy-proof guarantees under voters' partial information.
We use the notion of \emph{Risk-Avoiding Truthfulness} (RAT) and its \emph{RAT-Degree} proposed by Hartman, Segal-Halevi, and Tao (EC'25), the number of other voters whose information a manipulator must know before a safe and profitable deviation is possible.
Standard proportional rules, such as \emph{Greedy Approval Voting} (GreedyAV), \emph{Proportional Approval Voting} (PAV), and the \emph{Method of Equal Shares} (MES), have poor performances under the RAT-Degree metric (with low RAT-Degrees).

Finally, for $n$ voters, we construct a JR mechanism with RAT-Degree at least $n-5$, without a polynomial-time bound.
We then use the same idea together with algebraic ballot encodings to obtain a polynomial-time JR mechanism with RAT-Degree at least $n-7$.
When voters also have immutable private identifiers and both the identifiers and ballots of unknown voters are unobserved, we obtain a polynomial-time JR mechanism with RAT-Degree at least $n-1$ for $2\le k<m$.
This degree is exactly $n-1$ when $k\ge3$ and $k$ divides $n$.
\end{abstract}

\section{Introduction}

Approval-based multiwinner elections study the problem of selecting a committee of several candidates from a larger candidate set, where each voter reports the subset of candidates that she finds acceptable.
This model has become one of the central settings in computational social choice; see the book of Lackner and Skowron~\cite{lackner2023multi}.
Approval ballots are simple enough to be used in practice, while still allowing voters to express support for multiple alternatives.
Therefore, multiwinner approval rules are relevant in a wide range of applications, including parliamentary and committee elections, participatory decision making, shortlisting, group recommendation, and the selection of representative panels.
Unlike in single-winner elections, the purpose is not merely to identify the most popular alternative.
A good committee should also reflect the diversity of the electorate, so that large groups of voters with coherent interests are not completely ignored.

A fundamental way to formalize this idea is \emph{Justified Representation} (JR), introduced by Aziz et al.~\cite{aziz2017justified}.
JR captures a minimal proportionality requirement.
If a group of at least $n/k$ voters all approve some common candidate, then this group is large enough to deserve one seat in a committee of size $k$; hence, the selected committee should contain at least one candidate approved by at least one voter in the group.
Equivalently, no quota-sized cohesive group should be left entirely unrepresented.
The quota intuition behind this requirement is closely related to classical ideas in apportionment~\cite{pukelsheim2017quota}.
Although JR is deliberately weak, it serves as an important baseline: any rule that fails JR can be viewed as violating a very basic form of group representation.

The literature has developed many stronger and related representation axioms.
\emph{Proportional Justified Representation} (PJR) strengthens JR by requiring larger cohesive groups to receive multiple representatives~\cite{fernandez2017pjr}, while \emph{Extended Justified Representation} (EJR) asks that some voter in such a group personally obtain the appropriate level of representation; the complexity of checking and computing such guarantees has also been studied extensively~\cite{aziz2018ejr}.
More recent work has proposed and analyzed further refinements, including group-justifying axioms~\cite{elkind2023justifying}, robust and verifiable proportionality axioms~\cite{brill2023ejr+}, \emph{Individual Representation} (IR) guarantees~\cite{brill2022individual,brill2025individual}, and \emph{Fully Proportional Justified Representation} (FPJR)~\cite{KalayciLK2025}.
\emph{Average Justified Representation} (AJR) has also been considered, including recent work on the likelihood of its existence~\cite{han2026likelihood}.
Another line of work measures the extent to which rules satisfy proportionality through quantitative notions such as the \emph{Proportionality Degree} and the \emph{Degree of Justified Representation}~\cite{skowron2021proportionality,janeczko2022proportionality}.
Several well-studied voting rules, such as Phragmén's methods, were specifically analyzed from the perspective of Justified Representation \cite{brill2017phragmen,brill2024phragmen}.
Related proportionality ideas also arise in market-based and core-style approaches to committee selection \cite{PetersP0021,JiangMW20,MunagalaSWW22}, as well as in extensions with allocation constraints~\cite{MavrovMS23}.
Together, this body of work shows that Justified Representation and its variants have become a central language for reasoning about fairness in approval-based committee voting.

The preceding discussion assumes that voters report their approval preferences sincerely. 
In strategic settings, however, voters may benefit from misrepresentation. 
For instance, a voter may withhold approval from a popular candidate who is likely to be elected anyway, hoping to secure a seat for another approved candidate. 
Such free-riding is a well-known issue in \emph{Proportional Approval Voting} (PAV).

\emph{Strategyproofness} rules out such incentives by requiring truthful reporting to be a dominant strategy. 
This is desirable both normatively, since representation guarantees should not depend on sincere behavior despite incentives to manipulate, and practically, since it reduces strategic complexity and protects less sophisticated voters.

Unfortunately, proportionality and strategyproofness are difficult to reconcile in approval-based committee elections. 
Peters~\cite{peters2018proportionality} proved a strong incompatibility result: even weak proportionality and strategyproofness cannot coexist under an additional \emph{Weak Efficiency} requirement. 
Since JR implies the proportionality condition used in that theorem, the result helps explain the manipulability of standard proportional approval rules. 
Related incompatibility and structural results appear in the broader literature on proportionality and welfarism~\cite{peters2020proportionality}.

However, Peters' result leaves two questions open for our purposes. 
First, it concerns deterministic, resolute committee rules and therefore does not exclude randomized mechanisms that are strategyproof in expectation. 
Second, it assumes Weak Efficiency, which is logically independent of JR and strategyproofness. 
It is therefore natural to ask whether the incompatibility persists without this additional assumption.

\subsection{Our Results}

Our first contribution is to close these gaps.
We show that randomization does not restore compatibility between JR and strategyproofness.
More precisely, we prove that there is no randomized approval-based committee mechanism that satisfies \emph{Ex-Post JR} and \emph{Ex-Ante strategyproofness} (Theorem~\ref{thm:main-impossibility}).
Ex-Post JR requires every committee in the support of the lottery to satisfy JR, while Ex-Ante strategyproofness requires truthful reporting to maximize each voter's \emph{expected utility}.
The impossibility already appears in a very small instance with three voters, four candidates, and target committee size three.
It holds even when the rule may leave seats unfilled, and it extends beyond approval-count utility to every common strictly increasing cardinality-based utility function.
By adapting Peters' parameter reductions~\cite{peters2018proportionality}, we extend the base obstruction to every target size $k\geq3$, every electorate size divisible by $k$, and every candidate-set size at least $k+1$ (Corollary~\ref{cor:scaled-impossibility}).
Our result strengthens Peters' theorem in two aspects: firstly, it applies to randomized mechanisms; secondly, it does not require the mechanism to satisfy Weak Efficiency.

Our result is in sharp contrast to the mechanism design problem in other social choice areas such as \emph{fair division}.
In many settings studied in the fair division problem, introducing randomness in the mechanism can achieve both Ex-Ante strategyproofness and fairness~\cite{mossel2010truthful,dobzinski2013power,babaioff2022best,sun2025randomized,babaioff2026truthful,bu2025truthful}.

Our theorem allows the mechanism to select fewer than $k$ candidates, showing that the obstruction is not merely an artifact of forcing the rule to fill every seat.
In fair division literature, allowing ``resource burning'' (i.e., with a subset of resources unallocated) can be helpful to achieve strategyproofness~\cite{cole2013mechanism}.
However, our result shows that similar ideas in the multiwinner election problem cannot help.

This strong impossibility motivates the second part of the paper.
If exact strategyproofness is incompatible with JR, then one can ask for weaker incentive guarantees that better reflect the information available to voters.
We adopt the framework of \emph{Risk-Avoiding Truthfulness} (RAT) and the associated \emph{RAT-Degree} introduced by Hartman, Segal-Halevi, and Tao~\cite{hartman2025s}.
A manipulation is relevant in this framework only if it is safe under the manipulator's uncertainty: it must never make the manipulator worse off over all completions of the unknown voters' information, and it must make her strictly better off for at least one completion.
The RAT-Degree of a rule measures how many other voters' reports a manipulator must know before such a safe and profitable deviation becomes possible.
Thus, while classical strategyproofness asks whether manipulation is possible under full information, RAT-Degree measures the informational robustness of a rule against risk-avoiding manipulators.

We investigate the RAT-Degree of several voting rules satisfying JR. 
Specifically, we demonstrate that standard proportional rules including \emph{Greedy Approval Voting} (GreedyAV), \emph{Proportional Approval Voting} (PAV), and the \emph{Method of Equal Shares} (MES) having a low RAT-Degree. 
That is, there exist preference profiles in which a voter can safely and profitably manipulate the outcome while observing only a relatively small fraction of reports from other agents. 
In this regard, our findings complement existing impossibility results on full strategyproofness by providing a fine-grained analysis of standard JR rules under partial information.

On the positive side, we first establish that high resistance to manipulation is theoretically achievable by constructing a JR mechanism with a RAT-Degree of at least $n - 5$ (Theorem~\ref{thm:jr-high-rat}) without polynomial-time efficiency guarantees. 
We then combine a deterministic local-lemma construction for bounded candidate sets with algebraic ballot encodings to obtain a polynomial-time JR mechanism with RAT-Degree at least $n-7$ (Theorem~\ref{thm:jr-polynomial-rat}).

Finally, we consider an information model in which voters have distinct \emph{Immutable Identifiers} (IDs) that they cannot choose or change.
The mechanism observes all IDs, but a manipulator observes neither the IDs nor the ballots of unknown voters, and unknown IDs may be any positive integers distinct from the known IDs and one another.
For every $n\ge1$ and $2\le k<m$, we construct a JR mechanism that selects exactly $k$ candidates, runs in polynomial time in the total input length including the IDs, and has RAT-Degree at least $n-1$ under this information model (Theorem~\ref{thm:id-rat}).
When $k\ge3$ and $k$ divides $n$, its RAT-Degree is exactly $n-1$, so safe and profitable manipulation requires knowledge of every other voter's ID and ballot.

\subsection{Related Work}\label{sec:related-work}

Sinay and Gonen~\cite{sinay2026degree} independently study safe and profitable reports for PAV.
They consider only non-empty ballots different from the full candidate set, and the reported ballot must be a subset of the voter's true approval set.
Their safety condition requires the reported committee to contain every approved candidate selected under truth for every completion, whereas our utility is the number of approved candidates in the committee.
In their model, a safe and profitable report is possible with $\lceil n/k\rceil$ known ballots, and impossible with at most $\lfloor n/(k+1)\rfloor-1$ known ballots.
Under the conditions of Theorem~\ref{thm:pav-rat}, our PAV upper bound is also $\lceil n/k\rceil$.

\section{Preliminaries}

Let $\N=\{1,\ldots,n\}$ be the voter set, let $\C$ be a set of $m$ candidates, and let $1\le k\le m$ be the target committee size.
Each voter $i$ reports an approval ballot $A_i\subseteq\C$, and the resulting profile is $\calA=(A_1,\ldots,A_n)$.
Empty ballots and the full ballot $C$ are allowed.
We distinguish committees that fill all $k$ seats from those that may leave seats unfilled:
\[
 \calW_k=\{W\subseteq\C:|W|=k\},\qquad
 \calW_{\leq k}=\{W\subseteq\C:|W|\leq k\}.
\]
Unless another utility function is specified, a voter with approval set $A\subseteq\C$ has utility
\[
 u_A(W)=|W\cap A|
\]
from a committee $W$.
For a voter with true ballot $T$, all comparisons between reports use $u_T$, including when she reports a different ballot $R$.

\begin{definition}\label{def:JR}
A committee $W\subseteq\C$ satisfies Justified Representation at profile $\calA$ with target size $k$ if every group $X\subseteq\N$ with $|X|\ge n/k$ and $\bigcap_{i\in X}A_i\ne\varnothing$ has a represented member:
\[
 W\cap\bigcup_{i\in X}A_i\ne\varnothing.
\]
Let $\calJ(\calA)$ and $\calJ_{\leq k}(\calA)$ denote the JR committees in $\calW_k$ and $\calW_{\leq k}$, respectively.
\end{definition}

Adding candidates preserves JR, so every committee in $\calJ_{\leq k}(\calA)$ extends to one in $\calJ(\calA)$.
A rule satisfies \emph{Anonymity} if permuting the voter indices leaves its output unchanged.

We use the following notation to compare a truthful profile with a reported profile and to specify the ballots unknown to the manipulator.
\begin{definition}\label{def:profile-notation}
For a voter $i\in N$, a ballot $R\subseteq C$, and a set $J\subseteq N\setminus\{i\}$, define
\[
 \bigl(\calA[i\leftarrow R,J\leftarrow\varnothing]\bigr)_j=
 \begin{cases}
 R,&j=i,\\
 \varnothing,&j\in J,\\
 A_j,&j\notin J\cup\{i\}.
 \end{cases}
\]
Thus, voter $i$ reports $R$ and the ballots in $J$ become empty.
We write $\calA[i\leftarrow R]$ when $J=\varnothing$.

For $1\le s<n$, let $D$ contain $n-s-1$ known ballots and let $\bfU$ contain $s$ unknown ballots.
Write $D+T+\bfU$ for the profile formed from these ballots and the manipulator's ballot $T$, with the voter positions fixed when comparing $T$ with a report $R$.
The notation $\varnothing^s$ denotes $s$ empty ballots.
\end{definition}

\section{JR and Strategyproofness under Randomization}

\begin{definition}\label{def:randomized-rule}
A \emph{Randomized Committee Rule} $\calM$ assigns each profile $\calA$ a probability distribution over $\calW_{\leq k}$.
It satisfies \emph{Ex-Post JR} if every committee in the support of $\calM(\calA)$ belongs to $\calJ_{\leq k}(\calA)$.
The JR quota remains $n/k$ even when a committee leaves seats unfilled.

The rule satisfies \emph{Ex-Ante Strategyproofness} if, for every profile $\calA$, voter $i\in\N$, and report $R\subseteq\C$,
\[
 \E_{W\sim\calM(\calA)}[u_{A_i}(W)]
  \ge
 \E_{W\sim\calM(\calA[i\leftarrow R])}[u_{A_i}(W)].
\]
Both sides use the true ballot $A_i$ to compute utility, with the other ballots fixed.
\end{definition}

The incompatibility already occurs with three voters, four candidates, and three seats.
The proof only requires a voter who approves two candidates to prefer having both selected to having just one selected.

\begin{theorem}\label{thm:main-impossibility}
Let $n=k=3$ and $C=\{a,b,c,d\}$.
Suppose that all voters have utility $u_T(W)=f(|W\cap T|)$ for the same function $f:\{0,1,2,3\}\to\RR$, with $f(2)>f(1)$.
No randomized rule with range $\calW_{\leq3}$ satisfies both Ex-Post JR and Ex-Ante Strategyproofness.
\end{theorem}
\begin{proof}
Suppose that such a rule $\calM$ exists.
We first make the rule neutral by averaging over candidate permutations.
Strategyproofness will then force all four candidates to be selected at one profile.
Since $n/k=1$, JR requires every voter with a nonempty ballot to be represented in every committee in the support.
In particular, a voter who approves two candidates has truthful utility $f(1)$ or $f(2)$.

Let $\frakS_C$ be the set of all candidate permutations of $C$.
For $\sigma\in\frakS_C$ and $S\subseteq C$, write $\sigma(S)=\{\sigma(c):c\in S\}$, and put $\sigma\calA=(\sigma(A_1),\sigma(A_2),\sigma(A_3))$.
Voter indices remain fixed throughout this averaging.
Since there are $24$ candidate permutations, define
\[
 \Pr[\overline{\calM}(\calA)=W]
 =\frac1{24}\sum_{\sigma\in\frakS_C}
 \Pr[\calM(\sigma\calA)=\sigma(W)]
 \qquad(W\in\calW_{\leq3}).
\]
Candidate permutations preserve JR and committee size, so $\overline{\calM}$ satisfies Ex-Post JR and selects at most three candidates.
They also preserve utility, since $|\sigma(W)\cap\sigma(T_i)|=|W\cap T_i|$.
For each $\sigma$ and report $R_i$, apply strategyproofness of $\calM$ to the true ballot $\sigma(T_i)$ and the report $\sigma(R_i)$.
Averaging these inequalities proves Ex-Ante Strategyproofness of $\overline{\calM}$.
The uniform average satisfies \emph{Neutrality}: permuting candidates in every ballot permutes the output distribution in the same way.

Consider distinct voters $i,j,h$ with ballots $\{x\}$, $\{y,z\}$, and $\{y,w\}$, respectively, where $x,y,z,w$ are distinct candidates.
We show that $\overline{\calM}$ selects $\{x,y,z\}$ and $\{x,y,w\}$ with probability $1/2$ each.
JR requires $x$ to be selected, together with $y$ or both $z,w$.
Since committee size is at most $3$, the only possible committees are
\[
 \{x,y\},\quad \{x,y,z\},\quad \{x,y,w\},\quad \{x,z,w\}.
\]
Let $\alpha$ and $\beta$ be the probabilities of $\{x,y,z\}$ and $\{x,y,w\}$, respectively, so $\alpha+\beta\le1$.
We use the strategyproofness condition for each of voters $j$ and $h$ to compute these probabilities.

Voter $j$ approves $\{y,z\}$, and only $\{x,y,z\}$ contains both candidates.
Her truthful expected utility is therefore $(1-\alpha)f(1)+\alpha f(2)$.
If she reports $\{z\}$, the ballots of voters $i,j,h$ become $\{x\}$, $\{z\}$, and $\{y,w\}$.
JR and committee size at most $3$ allow only $\{x,y,z\}$ and $\{x,z,w\}$.
The permutation $y\leftrightarrow w$ fixes every ballot and exchanges these two committees.
Hence, neutrality gives each committee probability $1/2$.
Only the first committee contains both of her true approved candidates, so Ex-Ante Strategyproofness requires
\[
 (1-\alpha)f(1)+\alpha f(2)
 \ge\tfrac12 f(1)+\tfrac12 f(2).
\]
Since $f(2)>f(1)$, we obtain $\alpha\ge1/2$.

For voter $h$, the report $\{w\}$ allows only the JR committees $\{x,y,w\}$ and $\{x,z,w\}$.
The permutation $y\leftrightarrow z$ fixes every reported ballot and exchanges these committees, so each has probability $1/2$.
The same strategyproofness argument gives $\beta\ge1/2$.
Together with $\alpha+\beta\le1$, these bounds imply $\alpha=\beta=1/2$, leaving probability zero for the other committees.
The argument holds for every choice of the distinct voter indices $i,j,h$.

Now consider the truthful profile
\[
 \calB=(T_1,T_2,T_3)=(\{a,b\},\{a,c\},\{a,d\}).
\]
When voter $i$ reports $T_i\setminus\{a\}$, the preceding result applies with $y=a$ and $\{x\}=T_i\setminus\{a\}$.
Both possible committees contain $\{x,y\}=T_i$, so this report gives utility $f(2)$ with probability one.
JR gives truthful utility $f(1)$ or $f(2)$.
Since $f(2)>f(1)$, Ex-Ante Strategyproofness forces the truthful committee to contain $T_i$ with probability one.
All three requirements hold for the same distribution $\overline{\calM}(\calB)$.
Thus, every committee in its support must contain $T_1\cup T_2\cup T_3=C$, contradicting committee size at most $3$.
\end{proof}

To extend the impossibility to larger elections, we copy each ballot and fill the additional seats with candidates required by JR.
We use the parameter reductions of Peters~\cite{peters2018proportionality}.

\begin{corollary}\label{cor:scaled-impossibility}
Let $k\ge3$, $q\in\ZZ^+$, $n=qk$, and $|C|\ge k+1$.
Suppose that all voters have utility $u_T(W)=f(|W\cap T|)$ for the same function $f:\{0,\ldots,k\}\to\RR$, with $f(2)>f(1)$.
No randomized rule with range $\calW_{\leq k}$ satisfies both Ex-Post JR and Ex-Ante Strategyproofness.
The same result holds for range $\calW_k$.
\end{corollary}
\begin{proof}
Suppose that such a rule $\calM$ exists.
We use $\calM$ to construct a rule for three voters and four candidates, contradicting Theorem~\ref{thm:main-impossibility}.
Choose $C_0=\{a,b,c,d\}\subseteq C$ and $Z\subseteq C\setminus C_0$ with $|Z|=k-3$, which is possible because $|C|\ge k+1$.
For a profile $(T_1,T_2,T_3)$ over $C_0$, write $T^q$ for $q$ copies of a ballot $T$ and form
\[
 \calA^\star=T_1^q+T_2^q+T_3^q+\sum_{z\in Z}\{z\}^q,
\]
using fixed voter positions for each group of copies.
Apply $\calM$ to $\calA^\star$ and replace each output $W$ by $W\cap C_0$.

The profile has $3q+(k-3)q=qk=n$ voters and quota $q$.
For every $z\in Z$, the $q$ singleton ballots $\{z\}$ force $z$ into every output by JR, so $|W\cap C_0|\le k-(k-3)=3$.
For each nonempty $T_i$, its $q$ copies form a cohesive group, so $W\cap C_0$ contains a candidate of $T_i$.
Thus, the resulting rule satisfies Ex-Post JR for $3$ voters and target size $3$.
Projection preserves the utility of each of the $3$ voters because $T_i\subseteq C_0$ gives $f(|W\cap T_i|)=f(|(W\cap C_0)\cap T_i|)$.

To prove Ex-Ante Strategyproofness, suppose that one of these voters with true ballot $T$ has a profitable report $R$.
Change her $q$ copies from $T$ to $R$ one at a time, and let $U_t$ be the expected utility for true ballot $T$ after the first $t$ copies report $R$, for $0\le t\le q$.
Since the projected utility is unchanged, profitability gives $U_q>U_0$, so some step has $U_t>U_{t-1}$.
At that step, the copy changing from $T$ to $R$ has a profitable report under $\calM$, with every other report fixed, a contradiction.
This gives the required contradiction.
\end{proof}

\section{RAT-Degree of GreedyAV, PAV, and MES}\label{sec:standard-rat}

We now consider voters who know only some of the other ballots.
A report is useful to a risk-avoiding voter if it is safe at every completion of the unknown ballots and profitable at some completion.
RAT-Degree~\cite{hartman2025s} measures how many ballots must be known before such a report exists.
Throughout this section, $\mu$ maps profiles to committees in $\calW_k$.

\begin{definition}\label{def:safe-profitable}
Fix a voter $i$ with true ballot $T$ and known ballots for voters in $K\subseteq N\setminus\{i\}$.
For each completion of the unknown ballots, let $\calA$ be the resulting profile with $A_i=T$.
A ballot $R\ne T$ is a \emph{Safe Report} if, for every completion,
\[
 u_T\bigl(\mu(\calA[i\leftarrow R])\bigr)\ge u_T\bigl(\mu(\calA)\bigr).
\]
It is a \emph{Profitable Report} if the inequality is strict for at least one completion.
\end{definition}

\begin{definition}
The RAT-Degree $\rd(\mu)$ is the minimum $|K|$ for which a safe and profitable report exists in Definition~\ref{def:safe-profitable}.
If no such report exists, set $\rd(\mu)=n$.
\end{definition}

\subsection{Greedy Approval Voting}

GreedyAV selects candidates in $k$ steps, counting only approvals from voters not yet represented.
Start with an empty committee and call all voters uncovered.
Fix a candidate order $c_1\prec\cdots\prec c_m$ for ties.
At each step, select a candidate outside the committee with the most approvals from uncovered voters, and mark the approving voters as covered.
Continue until $k$ candidates have been selected, even if all remaining scores are zero.
To see why GreedyAV satisfies JR, suppose that a cohesive group is not represented.
Its common candidate has uncovered score at least $n/k$ at every step, so at least $n/k$ voters must be first covered at each of the $k$ steps.
Every voter would then be represented, a contradiction.

For $n\ge k+1$, our manipulation uses a true ballot $T=C\setminus\{x\}$.
Every committee gives utility $k-1$ or $k$, so safety reduces to one condition: whenever the truthful outcome omits $x$, the reported outcome must also omit $x$.
\begin{theorem}\label{thm:greedy-rat}
For $n\ge2$ and $2\le k<m$, GreedyAV satisfies
\[
 \rd(\GreedyAV)\le\left\lfloor\frac{n-1}{k+1}\right\rfloor
 =\left\lceil\frac{n}{k+1}\right\rceil-1.
\]
\end{theorem}
\begin{proof}
Write
\[
 d=\left\lfloor\frac{n-1}{k+1}\right\rfloor,\qquad
 n=(k+1)d+r,\qquad 1\le r\le k+1.
\]
For $n\ge k+1$, the manipulator has true ballot $T=C\setminus\{x\}$ and $d$ known ballots $\{x\}$.
We choose $x$ and the report according to $r$.

A truthful outcome without $x$ requires each selected candidate to defeat the $d$ known approvals of $x$.
The score needed depends on the fixed order.
For $c\ne x$, define
\[
 \tau(c)=
 \begin{cases}
 d,&c\prec x,\\
 d+1,&x\prec c.
 \end{cases}
\]
If the truthful outcome does not contain $x$, let $w_1,\ldots,w_k$ be the candidates in the order in which GreedyAV selects them and let $s_j$ be the uncovered score of $w_j$ when selected.
The $d$ known voters remain uncovered whenever $x$ is not selected, so the score of $x$ is always at least $d$.
The fixed order then gives $s_j\ge\tau(w_j)$, including choices with score zero.
Each $s_j$ counts the voters first covered by $w_j$.
These sets of voters are disjoint and do not contain any of the $d$ known voters, so
\begin{equation}\label{eq:greedy-counting}
 s_j\ge\tau(w_j)\quad(1\le j\le k),\qquad
 \sum_{j=1}^k s_j\le n-d.
\end{equation}
Moreover, the manipulator approves every $w_j$ and is covered by $w_1$.
For $j\ge2$, her initial approval for $w_j$ is in addition to the $s_j$ voters still uncovered when $w_j$ is selected, so the initial score of $w_j$ is at least $s_j+1$.

If $r=k+1$, put
\[
 x=c_1,\qquad R=\{c_2,\ldots,c_{m-k+1}\},\qquad
 B=T\setminus R,\qquad y=c_2.
\]
The set $B$ has size $k-1$, and every member of the nonempty set $R$ is before every member of $B$ in the fixed order.
To prove safety, fix a completion for which the truthful committee does not contain $x$.
Since $x$ is first in the order, \eqref{eq:greedy-counting} gives $s_j\ge d+1$ for each selected candidate, while
\[
 n-d=k(d+1)+1.
\]
If the first selected candidate $c$ belongs to $B$, then $|B|=k-1$ requires some candidate $z\in R$ to be selected later.
For $z=w_j$ with $j\ge2$, the initial score is at least $s_j+1\ge d+2$.
Because $z\prec c$, selecting $c$ first requires its score to be at least $d+3$.
Hence,
\[
 \sum_{j=1}^k s_j\ge(d+3)+(k-1)(d+1)=k(d+1)+2>n-d,
\]
contradicting~\eqref{eq:greedy-counting}.
The first truthful choice therefore belongs to $R$.
Changing $T$ to $R$ does not change its score and can only reduce other scores, so both runs select the same first candidate.
The manipulator is then covered, leaving identical uncovered ballots and hence identical remaining choices.
If the truthful committee contains $x$, its utility is already the minimum possible value $k-1$.
Thus, the report is safe.

For profitability, choose the unknown ballots as follows:
\[
\begin{array}{c|c}
 \text{Ballot}&\text{Number of copies}\\\hline
 \{y\}&d\\
 \{b\},\quad b\in B&d+1\text{ for each }b\\
 \varnothing&1
\end{array}
\]
The total number is $d+(k-1)(d+1)+1=n-d-1$.
Under truth, each $b\in B$ has initial score $d+2$, while $y$ has score $d+1$ and every other member of $R$ has score one.
A member of $B$ is selected first, so the manipulator is covered.
The remaining members of $B$ have score $d+1$, while $x$ and $y$ have score $d$.
After all $k-1$ members of $B$ are selected, the fixed order gives the last seat to $x$.
Under $R$, candidate $y$ and all members of $B$ have initial score $d+1$, and $y$ is first in the fixed order among the candidates of maximum score.
Once $y$ is selected, the manipulator is covered, and the remaining $k-1$ choices are the members of $B$, each with score $d+1$.
The truthful committee is $B\cup\{x\}$ and the reported committee is $B\cup\{y\}$, giving true utility $k-1$ and $k$, respectively.

Now suppose that $1\le r\le k$.
Since $n\ge k+1$, we have $d\ge1$.
Put
\[
 a=k-r,\qquad x=c_{a+1},\qquad
 P=\{c_1,\ldots,c_{k+1}\},\qquad B=P\setminus\{x\}.
\]
Let $y$ be the first member of $B$ in the fixed order, and set $R=\{y\}\cup(C\setminus P)$.
Exactly $a$ candidates are before $x$, and $|B|=k$.
We show that every truthful outcome contains $x$, which gives safety because its utility is then $k-1$.
Suppose otherwise, and let $w_1,\ldots,w_k$ be the candidates in the order in which GreedyAV selects them under truth.
At most $a$ of these candidates are before $x$, so
\begin{equation}\label{eq:greedy-thresholds}
 \sum_{j=1}^k\tau(w_j)\ge ad+(k-a)(d+1)=kd+r=n-d.
\end{equation}
Let $c=w_1$ and $z=w_2$.
The initial score of $z$ is at least $s_2+1\ge\tau(z)+1$, because the manipulator approves $z$ and is covered by $c$.
We claim that $s_1\ge\tau(c)+1$ as well.
If $c\prec x$, then $s_1$ is at least the initial score of $z$, which is at least $d+1=\tau(c)+1$.
If $x\prec c$ and $x\prec z$, the initial score of $z$ is at least $d+2=\tau(c)+1$.
Finally, if $z\prec x\prec c$, then the initial score of $z$ is at least $d+1$, and $z\prec c$ requires $s_1\ge d+2=\tau(c)+1$.
Combining these inequalities with~\eqref{eq:greedy-thresholds} gives
\[
 n-d\ge\sum_{j=1}^k s_j
 \ge1+\sum_{j=1}^k\tau(w_j)\ge n-d+1,
\]
a contradiction.
Thus, the truthful committee always contains $x$, and $R$ is safe.

For profitability, use $\tau(y)-1$ copies of $\{y\}$ and $\tau(v)$ copies of $\{v\}$ for each $v\in B\setminus\{y\}$ as the unknown ballots.
Since $d\ge1$, we have $\tau(y)-1\ge0$ and $\tau(v)\ge0$ for each $v\in B\setminus\{y\}$.
Exactly $a$ members of $B$ are before $x$ and $r$ are after $x$, giving a total of
\[
 \sum_{v\in B}\tau(v)-1=ad+r(d+1)-1=n-d-1
\]
unknown ballots.
At least one member of $B\setminus\{y\}$ is after $x$.
Indeed, if $a>0$, then $y\prec x$ and $r\ge1$ members of $B$ are after $x$; if $a=0$, all $k\ge2$ members of $B$ are after $x$.
Under truth, these members have initial score $d+2$, larger than the scores of $x$, $y$, and every candidate outside $P$.
One of these members is selected first, so the manipulator is covered.
Each remaining $v\in B\setminus\{y\}$ then has score $\tau(v)$ and is selected before $x$, whose score is $d$.
Candidate $y$ has score $\tau(y)-1$ and cannot be selected before $x$, while candidates outside $P$ have score zero.
The truthful outcome is therefore $(B\setminus\{y\})\cup\{x\}$.

Under $R$, only $y$ or a candidate outside $P$ can cover the manipulator.
Before $y$ is selected, its score is $\tau(y)\ge1$, while each candidate outside $P$ has at most one approval and is after $y$ in the fixed order.
Thus, $y$ is selected before any candidate outside $P$.
Each member $v\in B$ has score $\tau(v)$ before being selected, so every member of $B$ is selected before $x$.
The reported committee is $B$, giving true utility $k$ instead of $k-1$.

It remains to consider $2\le n\le k$.
Take $p=c_1$, $q=c_{k+1}$, $T=\{p,q\}$, and $R=\{q\}$, with no known ballots.
Under $R$, candidate $q$ must be selected within the first $n\le k$ steps.
Before $q$ is selected, the manipulator remains uncovered and gives it positive support.
Each earlier choice therefore has positive score, and at least one of the other $n-1$ voters is first covered by that choice.
Thus, the reported outcome always gives true utility at least one.
For safety, it suffices to consider a completion where the truthful outcome contains both $p$ and $q$.
Candidate $q=c_{k+1}$ cannot be selected with score zero: if $t$ candidates have already been selected, at least $k-t$ of the candidates before $q$ remain, enough for the remaining seats.
Moreover, some voter other than the manipulator must approve $q$.
Otherwise, before $p$ is selected, $p$ has at least the same score as $q$ and is first in the fixed order; after $p$ is selected, $q$ has score zero and cannot be selected.
Under $R$, the initial score of $q$ is therefore at least two, so at least two voters are covered by the first choice.
Since at least one further voter is covered by each later choice with positive score, there are at most $n-1\le k-1$ choices with positive score.
At least one candidate is therefore selected with score zero.
Candidate $p=c_1$ is selected before or at that step, and $q$ is selected by the argument above.
Hence, $R$ is safe.
For profitability, let all $n-1$ other voters report $\{p\}$.
Under truth, $p$ is selected first and all voters are covered, giving committee $\{c_1,\ldots,c_k\}$ and utility one.
Under $R$, $p$ is again selected first, and $q$ is selected next using the manipulator's approval, giving utility two.
Each case gives a safe and profitable report with $d$ known ballots.
\end{proof}

\subsection{Proportional Approval Voting}

The PAV score assigns a smaller marginal contribution to each additional committee member approved by a voter.
Let $H_0=0$ and $H_t=\sum_{j=1}^t1/j$ for $t\ge1$.
The score of a committee $W$ is
\[
 \PAV(W)=\sum_{i\in N}H_{|W\cap A_i|}.
\]
PAV selects a size-$k$ committee of maximum score, using a fixed committee order for ties.
PAV satisfies EJR and hence JR~\cite{aziz2017justified,lackner2023multi}.

The manipulator removes approval from a candidate whose selection is already required by the known ballots and JR.
This reduces the PAV score less for committees containing her other approved candidate, which makes the report safe.

\begin{theorem}\label{thm:pav-rat}
Let $q=\lceil n/k\rceil$ and $r=\lceil k/2\rceil$.
If $k\ge4$, $|C|\ge k+1$, and $n\ge q+r+1$, then
\[
 \rd(\PAV)\le q.
\]
\end{theorem}
\begin{proof}
Choose distinct candidates $x,y$.
The manipulator has true ballot $\{x,y\}$, reports $\{y\}$, and has $q$ known voters who approve only $x$.
JR requires $x$ in every PAV outcome under both reports, so the manipulator always obtains utility at least one.
For safety, it remains to consider a truthful winner $W$ containing both $x,y$.
Changing the ballot reduces the score of $W$ by $1/2$, but reduces the score of any committee $V$ containing $x$ and omitting $y$ by $1$.
Since $\PAV(W)\ge\PAV(V)$ under truth, $W$ has strictly larger score than $V$ under the report.
Every reported winner therefore contains $y$ as well, proving safety.

For profitability, choose a set $S$ of $k-1$ candidates outside $\{x,y\}$ and put $C'=\{x,y\}\cup S$.
Let $r$ unknown voters approve exactly $S$, and let all remaining unknown ballots be empty.
The condition $n\ge q+r+1$ gives enough unknown voters.
Every candidate in $C'$ is approved by some voter, and no candidate outside $C'$ is approved.
Replacing a candidate outside $C'$ by an omitted candidate in $C'$ therefore gives a strictly larger PAV score under both reports.
Hence, every winning committee has the form $W_c=C'\setminus\{c\}$ for some $c\in C'$.
The scores are
\[
\begin{array}{c|c|c}
 \text{Committee}&\text{Truth}&\text{Report}\\\hline
 W_y&q+1+rH_{k-1}&q+rH_{k-1}\\
 W_x&1+rH_{k-1}&1+rH_{k-1}\\
 W_c,\ c\in S&q+\frac32+rH_{k-2}&q+1+rH_{k-2}
\end{array}
\]
Our choice $r=\lceil k/2\rceil$, together with $k\ge4$, gives
\[
 \frac12<\frac{r}{k-1}<1\le q.
\]
Under truth, the score of $W_y$ is larger than that of $W_x$ by $q$ and that of each $W_c$ with $c\in S$ by $r/(k-1)-1/2$, so $W_y$ is the unique winner.
Under the report, the score of each $W_c$ with $c\in S$ is larger than that of $W_y$ by $1-r/(k-1)$ and that of $W_x$ by $q-r/(k-1)$.
Thus, the reported winners all contain both $x,y$, giving true utility two instead of one.
\end{proof}

\subsection{The Method of Equal Shares}

MES was introduced for approval committees as \emph{Rule X} by Peters and Skowron~\cite{peters2020proportionality} and extended to participatory budgeting by Peters et al.~\cite{peters2021proportional}.
MES selects candidates using equal initial budgets $b=k/n$ for the voters.
Write $b_i$ for voter $i$'s remaining budget.
A candidate $c$ outside the committee is affordable if there exists $\rho\ge0$ such that
\[
 \sum_{i:c\in A_i}\min\{b_i,\rho\}\ge1.
\]
Let $\rho(c)$ be the least such value.
While a candidate is affordable and the committee has size below $k$, select an affordable candidate with minimum $\rho(c)$ and charge each approving voter $i$ $\min\{b_i,\rho(c)\}$.
Fix a candidate order $c_1\prec\cdots\prec c_m$ for ties and use the same order to add candidates if the committee still has size below $k$ when no candidate is affordable.
We call the steps that charge voters the MES phase.

The MES phase satisfies JR.
An unrepresented group of at least $n/k$ voters retains a total budget of at least one for a common approved candidate, which therefore remains affordable.
The phase cannot fill all $k$ seats while this budget remains unspent, nor can it stop for lack of an affordable candidate.

With one known empty ballot, some budget cannot be spent, so the rule must complete the committee using the fixed order.
The first candidate in that order is then selected whether or not the manipulator approves that candidate.

\begin{theorem}\label{thm:mes-rat}
If $n\ge k\ge2$ and $|C|\ge k+1$, then
\[
 \rd(\MES)\le1.
\]
\end{theorem}
\begin{proof}
Set $x=c_1$, $z=c_2$, $T=C\setminus\{z\}$, and $R=C\setminus\{x,z\}$.
The manipulator has one known empty ballot and $n-2$ unknown voters.
We first show that the report is safe.
The known empty ballot has unused budget $b$, leaving at most $(n-1)b<k$ for the MES phase.
Since every selected candidate requires one unit of budget, this phase selects at most $k-1$ candidates under both reports.
The final committee therefore contains $x$: if the MES phase does not select $x$, the fixed order selects $x$ first.
Moreover, the final committee omits $z$ exactly when the MES phase selects $k-1$ candidates outside $\{x,z\}$.
Indeed, if $z$ is not selected in the MES phase, the fixed order selects $z$ whenever $x$ is already selected or at least two seats remain; otherwise, only $x$ is added.

We next show that both MES phases give the same outcome whenever the truthful phase does not select $x$.
Suppose that the two phases have made the same earlier choices, with $x$ not among them.
The remaining budgets are identical because the reports differ only on $x$.
Each candidate other than $x$ is affordable under both reports with the same value of $\rho$, or is not affordable under both.
Candidate $x$ can be affordable under $R$ only if it is affordable under truth.
When $x$ is affordable under both reports, its value of $\rho$ cannot be lower under $R$.
If the truthful phase selects $c\ne x$, then $\rho(c)<\rho(x)$ whenever $x$ is affordable, since $x$ is first in the fixed order.
The reported phase therefore selects the same $c$.
If no candidate is affordable under truth, no candidate is affordable under the report.
By induction, both phases make the same choices whenever the truthful phase does not select $x$.

If the truthful committee contains $z$, its utility is already the minimum $k-1$.
Otherwise, the truthful MES phase selects $k-1$ candidates and omits both $x,z$.
The reported phase then selects the same candidates, and the fixed order selects $x$ under both reports.
Both reports therefore give utility $k$, proving safety.

For profitability, let every unknown voter report $C$, and put $h=n-1$.
These $h$ voters, including the manipulator, have initial budget $b$ each.
Under truth, each candidate outside $\{z\}$ has all $h$ approvals; under $R$, each candidate outside $\{x,z\}$ has all $h$ approvals.
The remaining candidates have $h-1$ approvals.
Since $m\ge k+1$, both reports give at least $k-1$ candidates with all $h$ approvals.
We show that the MES phase selects exactly $k-1$ of these candidates.
After $t\le k-2$ such choices, each of the $h$ voters has remaining budget $b-t/h$, and
\[
 h\left(b-\frac{t}{h}\right)=hb-t
 \ge hb-(k-2)=2-\frac{k}{n}\ge1.
\]
A candidate with all $h$ approvals therefore has $\rho=1/h$.
At that value, a candidate with only $h-1$ approvals receives at most $(h-1)/h<1$, so it cannot be selected first.
After $k-1$ choices, the remaining budget of the $h$ voters is
\[
 hb-(k-1)=1-\frac{k}{n}<1,
\]
so no further candidate is affordable.
Under truth, the MES phase includes $x$ and omits $z$, and the fixed order selects $z$.
Under $R$, the MES phase omits both $x,z$, and the fixed order selects $x$.
The report gives utility $k$ instead of $k-1$, proving profitability.
\end{proof}
\section{JR Mechanisms with High RAT-Degree}\label{sec:main}

We now construct two anonymous JR mechanisms that select exactly $k$ candidates and have RAT-Degree $n-O(1)$.
The finite construction in Sect.~\ref{sec:jr-finite} gives the bound $n-5$ without a polynomial-time guarantee.
Section~\ref{sec:jr-polynomial} gives a polynomial-time construction with the bound $n-7$.
When $k=1$, we use GreedyAV, which reduces to approval voting with a fixed candidate order; when $k=m$, the outcome can only be the whole candidate set.
Both cases satisfy JR and have RAT-Degree $n$, so we describe the constructions for $2\le k<m$.

\subsection{A Finite Construction}\label{sec:jr-finite}

We first assign each profile a nonempty family of JR committees such that every potentially profitable report also admits a strict loss at an all-empty completion.
We then choose one committee at each profile so that every such report has a strict-loss completion.
To construct these families, put $q=\lceil n/k\rceil$ and consider candidates supported by disjoint groups of $q$ approving voters.

\begin{definition}
A set $S\subseteq C$ is a \emph{Fundable Set} at profile $\calA$ if each candidate in $S$ can be assigned $q$ distinct approving voters, with no voter assigned twice.
Write $\fund(\calA)$ for the family of all fundable sets at $\calA$.
\end{definition}
Every fundable set has at most $\lfloor n/q\rfloor\le k$ candidates.

\begin{definition}
The \emph{Support Set} of $B\subseteq C$ at profile $\calA$ is the set of voters approving at least one candidate in $B$:
\[
 \supp_{\calA}(B)=\{i\in N:A_i\cap B\ne\varnothing\}.
\]
\end{definition}

To compare fundable families before and after a report changes, we use the following matching characterization.
\begin{lemma}\label{lem:hall}
$S\in\fund(\calA)$ if and only if $|\supp_{\calA}(B)|\ge q|B|$ for every subset $B\subseteq S$.
\end{lemma}
\begin{proof}
Define the bipartite graph $G=(X,Y,E)$ by
\[
 X=S\times[q],\qquad Y=N,\qquad
 E=\{((c,\ell),i)\in X\times Y:c\in A_i\},
\]
where $[q]=\{1,\ldots,q\}$ gives $q$ copies of each candidate.
For $X'\subseteq X$, let $N_G(X')$ be its set of neighbors and let $B=\{c\in S:(c,\ell)\in X'\text{ for some }\ell\}$.
Then,
\[
 N_G(X')=\supp_{\calA}(B),\qquad |X'|\le q|B|.
\]
Thus, the support inequalities imply $|N_G(X')|\ge|X'|$ for every $X'\subseteq X$, which is Hall's condition for a matching covering $X$.
Conversely, this condition applied to $X'=B\times[q]$ gives $|\supp_{\calA}(B)|\ge q|B|$ for every $B\subseteq S$.
By \emph{Hall's Theorem}~\cite{hall1935}, such a matching exists exactly when the condition holds, and these matchings correspond exactly to assignments of $q$ distinct approving voters to each candidate with no voter used twice.
\end{proof}

We will impose candidate requirements while keeping maximal fundable extensions available.
For this purpose, consider the candidates that can be added to every fundable set:
\[
 \partial\calA=\{c\in C:S\cup\{c\}\in\fund(\calA)\text{ for every }S\in\fund(\calA)\}.
\]

\begin{lemma}\label{lem:coloop}
We have $\partial\calA\in\fund(\calA)$, $|\partial\calA|\le k$, and
\[
 \partial\calA=
 \bigcap_{\substack{Q\in\fund(\calA)\\Q\text{ inclusion-maximal}}}Q.
\]
\end{lemma}
\begin{proof}
Adding the members of $\partial\calA$ to the empty set one at a time gives a fundable set, so $|\partial\calA|\le k$.
By maximality, every inclusion-maximal fundable set contains $\partial\calA$.
If $c\notin\partial\calA$, choose $S\in\fund(\calA)$ with $S\cup\{c\}\notin\fund(\calA)$ and extend $S$ to an inclusion-maximal fundable set $Q$.
The set $Q$ cannot contain $c$, since every subset of a fundable set is fundable.
This proves the intersection equality.
\end{proof}

\begin{lemma}\label{lem:extension}
Every inclusion-maximal fundable set $Q$ satisfies JR, and every fundable set is contained in such a set.
If the profile contains at least $s\ge1$ empty ballots, then
\[
 |Q|\le\left\lfloor\frac{n-s}{q}\right\rfloor\le k-1.
\]
\end{lemma}
\begin{proof}
Every fundable set has an inclusion-maximal fundable extension because $C$ is finite.
Let $Q$ be such a maximal set, with a matching assigning $q$ voters to each candidate.
If $Q$ fails JR, at least $q$ unrepresented voters approve a common candidate $c\notin Q$.
These voters approve no candidate in $Q$, so the matching does not use them.
Assigning $q$ of these voters to $c$ makes $Q\cup\{c\}$ fundable, contradicting maximality.
If at least $s\ge1$ ballots are empty, their voters cannot be used in the matching, so
\[
 q|Q|\le n-s<n\le qk,
\]
which gives the required size bound.
Any size-$k$ extension $W$ of $Q$ satisfies JR, although $W$ need not be fundable.
\end{proof}

\begin{lemma}\label{lem:fundable-avoidance}
Let $F\subseteq\partial\calA$.
If a candidate $c$ belongs to every committee in $\calJ(\calA)$ containing $F$, then $c\in\partial\calA$.
\end{lemma}
\begin{proof}
If $c\notin\partial\calA$, then Lemma~\ref{lem:coloop} gives an inclusion-maximal fundable set $Q$ that does not contain $c$.
The set $Q$ contains $F$ and satisfies JR by Lemma~\ref{lem:extension}.
Since $m>k$, extend $Q$ to a size-$k$ committee without $c$.
This committee still satisfies JR and contains $F$, contradicting the assumption.
\end{proof}

\begin{lemma}\label{lem:fundable-replacement}
Let $T=A_j$ and $\calA'=\calA[j\leftarrow R]$.
If $T\subseteq\partial\calA'$, then
\[
 T\subseteq\partial\calA
 \qquad\text{and}\qquad
 \fund(\calA)\subseteq\fund(\calA').
\]
\end{lemma}
\begin{proof}
Fix $S\in\fund(\calA)$.
The matching for $S\setminus T$ does not use voter $j$, whose ballot is $T$, so $S\setminus T$ remains fundable at $\calA'$.
Since $T\subseteq\partial\calA'$, add the members of $T$ one at a time to obtain $S\cup T\in\fund(\calA')$.
In particular, $S\in\fund(\calA')$, proving the inclusion of fundable families.

We next apply Lemma~\ref{lem:hall} to show that $S\cup T$ is also fundable at $\calA$.
For any $B\subseteq S\cup T$ disjoint from $T$, we have $B\subseteq S$, so $|\supp_{\calA}(B)|\ge q|B|$.
If $B\cap T\ne\varnothing$, voter $j$ belongs to $\supp_{\calA}(B)$, and every other ballot is unchanged, giving
\[
 |\supp_{\calA}(B)|\ge|\supp_{\calA'}(B)|\ge q|B|.
\]
The second inequality follows from $S\cup T\in\fund(\calA')$.
Thus, $S\cup T\in\fund(\calA)$ for every $S\in\fund(\calA)$, which implies $T\subseteq\partial\calA$.
\end{proof}

We now construct the committee families by adding requirements on which candidates must be selected.
If a report ensures the selection of every truly approved candidate when $s$ other ballots are empty, we require the truthful committees to contain those candidates as well.
Such a report then gives no utility gain.
The preceding lemmas ensure that these requirements never exclude all JR committees.

Fix an integer $1\le s<n$.
For every profile $\calA$, begin with $F_0(\calA)=\varnothing$ and define
\[
  \calK_r(\calA)=\{W\in\calJ(\calA):F_r(\calA)\subseteq W\}.
\]
For each $r\ge0$, define
\begin{equation}\label{eq:fundable-recurrence}
  \begin{split}
  F_{r+1}(\calA)=F_r(\calA)\cup\bigcup\bigl\{A_i:\;&\exists R\subseteq C,\ \exists J\subseteq N\setminus\{i\},\ |J|=s,\\
  &A_i\subseteq W\text{ for every }W\in\calK_r(\calA[i\leftarrow R,J\leftarrow\varnothing])\bigr\}.
  \end{split}
\end{equation}
All sets for $r+1$ are computed from the sets for $r$.

\begin{proposition}\label{prop:fundable-families}
The sequences $F_r(\calA)$ and $\calK_r(\calA)$ satisfy the following properties:
\begin{enumerate}
\item[(i)] For every $r\ge0$ and every profile $\calA$, we have
\[
  F_r(\calA)\subseteq F_{r+1}(\calA)
  \qquad\text{and}\qquad
  \calK_{r+1}(\calA)\subseteq\calK_r(\calA).
\]
There exists $r_0$ such that, for every $r\ge r_0$ and every profile $\calA$,
\[
  F_r(\calA)=F_{r_0}(\calA)
  \qquad\text{and}\qquad
  \calK_r(\calA)=\calK_{r_0}(\calA).
\]
These sets do not depend on the labels of the voters.
\item[(ii)] For every $r\ge0$, every profile $\calA$, any voter $j\in N$ with $A_j=\varnothing$, and any ballot $R\subseteq C$, let $\calA'=\calA[j\leftarrow R]$.
Then, $F_r(\calA)\subseteq F_r(\calA')$ and $\calK_r(\calA')\subseteq\calK_r(\calA)$.
\item[(iii)] For every $r\ge0$ and every profile $\calA$, we have $F_r(\calA)\subseteq\partial\calA$ and $\calK_r(\calA)\ne\varnothing$.
\end{enumerate}
\end{proposition}
\begin{proof}
\emph{(i)} By~\eqref{eq:fundable-recurrence}, we have $F_r(\calA)\subseteq F_{r+1}(\calA)$, and the definition of $\calK_r(\calA)$ then gives $\calK_{r+1}(\calA)\subseteq\calK_r(\calA)$.
There are $2^{mn}$ profiles and $m$ candidates, and the sets $F_r(\calA)$ only increase.
Therefore, after at most $m2^{mn}$ steps with an added candidate, all sets remain unchanged.
The families $\calK_r(\calA)$ then remain unchanged as well.
Every choice of a voter $i$, a report $R$, and a set $J$ is included in~\eqref{eq:fundable-recurrence}, independently of the labels of the voters.
The same is true of $\calJ(\calA)$, so the constructed sets do not depend on these labels.

\emph{(ii)} We prove (ii) by induction on $r$.
Fix a voter $j\in N$ with $A_j=\varnothing$, a ballot $R\subseteq C$, and the profile $\calA'=\calA[j\leftarrow R]$.
For $r=0$, the inclusion for $F_0$ holds because $F_0(\calA)=F_0(\calA')=\varnothing$.
Replacing an empty ballot with another ballot cannot remove a JR violation, so $\calJ(\calA')\subseteq\calJ(\calA)$.
Since $\calK_0(\calA)=\calJ(\calA)$, the inclusion for $\calK_0$ follows.
Assume that (ii) holds for $r$ at every profile.
Consider a ballot $A_i$ added to $F_{r+1}(\calA)$ by a choice of $R'$ and $J$ in~\eqref{eq:fundable-recurrence}.
If $i=j$, then $A_i=\varnothing$ does not add any candidates.
Otherwise, the ballot $A_i$ is unchanged at $\calA'$.
If $j\in J$, both modified profiles have an empty ballot at voter $j$, so
\[
  \calA'[i\leftarrow R',J\leftarrow\varnothing]
  =\calA[i\leftarrow R',J\leftarrow\varnothing].
\]
If $j\notin J$, the two modified profiles differ only in voter $j$'s ballot, from $\varnothing$ to $R$.
The assumed inclusion for step $r$ therefore gives
\[
  \calK_r(\calA'[i\leftarrow R',J\leftarrow\varnothing])
  \subseteq\calK_r(\calA[i\leftarrow R',J\leftarrow\varnothing]).
\]
In both cases, the same $R'$ and $J$ ensure that $A_i$ is added at $\calA'$.
Together with $F_r(\calA)\subseteq F_r(\calA')$, this proves $F_{r+1}(\calA)\subseteq F_{r+1}(\calA')$.
Combining this inclusion with $\calJ(\calA')\subseteq\calJ(\calA)$ gives $\calK_{r+1}(\calA')\subseteq\calK_{r+1}(\calA)$.

\emph{(iii)} We again use induction on $r$.
For $r=0$, we have $F_0(\calA)=\varnothing\subseteq\partial\calA$, and Lemma~\ref{lem:extension} gives $\calK_0(\calA)=\calJ(\calA)\ne\varnothing$.
Assume that (iii) holds for $r$ at every profile.
Consider any ballot $A_i$ added to $F_{r+1}(\calA)$ by $R$ and $J$.
Let $\calB=\calA[i\leftarrow R]$.
By (ii), replacing the empty ballots in $J$ with their ballots in $\calB$ gives
\[
  \calK_r(\calB)\subseteq\calK_r(\calA[i\leftarrow R,J\leftarrow\varnothing]).
\]
Hence, every committee in $\calK_r(\calB)$ contains $A_i$.
Since $F_r(\calB)\subseteq \partial\calB$, Lemma~\ref{lem:fundable-avoidance} gives $A_i\subseteq \partial\calB$.
Lemma~\ref{lem:fundable-replacement} then gives $A_i\subseteq \partial\calA$.
Thus, $F_{r+1}(\calA)\subseteq \partial\calA$.
Lemmas~\ref{lem:coloop} and~\ref{lem:extension} imply $\calK_{r+1}(\calA)\ne\varnothing$.
\end{proof}

At the fixed point, write $F(\calA)=F_{r_0}(\calA)$ and $\calK(\calA)=\calK_{r_0}(\calA)$.
The parameter $s$ remains fixed.
When at least $s$ ballots are already empty, we can simplify the recurrence by using those positions as $J$.

\begin{lemma}\label{lem:fundable-local}
If $\calA$ has at least $s$ empty ballots, then
\[
  F_{r+1}(\calA)=F_r(\calA)\cup\bigcup\{A_j:\exists R\subseteq C,\ A_j\subseteq W\text{ for every }W\in\calK_r(\calA[j\leftarrow R])\}.
\]
\end{lemma}
\begin{proof}
For each $A_j\ne\varnothing$, choose any $s$ voters with empty ballots as $J$ in~\eqref{eq:fundable-recurrence}.
This gives every addition in the required formula.
Conversely, for any choice of $R$ and $J$ in~\eqref{eq:fundable-recurrence}, Proposition~\ref{prop:fundable-families} gives
\[
  \calK_r(\calA[j\leftarrow R])\subseteq\calK_r(\calA[j\leftarrow R,J\leftarrow\varnothing]).
\]
Thus, every addition in~\eqref{eq:fundable-recurrence} also appears in the required formula.
Empty ballots do not add any candidates.
\end{proof}

We need many completions that leave the committee families unchanged.
The next lemma identifies ballots that can replace an empty ballot without changing $\calK(\calA)$.
Condition (ii) also permits some ballots with fewer than $k$ candidates, as needed when $m<2k-1$.

\begin{lemma}\label{lem:fundable-empty-ballot}
Suppose that $A_i=T$ and $|T|\ge m-k+1$, and that at least one of the following conditions holds:
\begin{enumerate}
\item[(i)] $|T|\ge k$;
\item[(ii)] $\calA$ has at least $s$ empty ballots and $C\setminus T\in\fund(\calA)$.
\end{enumerate}
Then, for every $r$,
\[
  F_r(\calA)=F_r(\calA[i\leftarrow\varnothing])
  \qquad\text{and}\qquad
  \calK_r(\calA)=\calK_r(\calA[i\leftarrow\varnothing]).
\]
\end{lemma}
\begin{proof}
Write $\calA^-=\calA[i\leftarrow\varnothing]$.
Since $|T|\ge m-k+1$, every size-$k$ committee represents voter $i$, so $\calJ(\calA)=\calJ(\calA^-)$.
The equalities therefore hold for $r=0$.
Assume that they hold for $r$ at every profile satisfying the lemma's assumptions.
Proposition~\ref{prop:fundable-families} gives $F_{r+1}(\calA^-)\subseteq F_{r+1}(\calA)$, so it suffices to prove the opposite inclusion.

First, suppose that (i) holds.
For any proposed $R$ and $J$ for voter $i$ in~\eqref{eq:fundable-recurrence}, write $\calB$ for the corresponding profile.
It has at least $s$ empty ballots.
By Proposition~\ref{prop:fundable-families} and Lemma~\ref{lem:extension}, there is a fundable set $B\in\calJ_{\leq k}(\calB)$ with $F_r(\calB)\subseteq B$ and $|B|\le k-1$.
Every size-$k$ committee containing $B$ belongs to $\calK_r(\calB)$.
Since $m>k$, the intersection of all these committees is exactly $B$.
They cannot all contain $T$, because $|B|<k\le|T|$.
Thus, voter $i$ cannot add $T$ at step $r+1$.
For any other voter $j$, the two profiles corresponding to a choice of $R$ and $J$ are identical if $i\in J$.
Otherwise, they differ only in voter $i$'s ballot, which is $T$ or empty.
The induction hypothesis under (i) gives equal families $\calK_r$ at these profiles.
Consequently, the same ballots are added at $\calA$ and $\calA^-$.

Next, suppose that (ii) holds, and write $O=C\setminus T$.
We cannot have $T\subseteq \partial\calA$.
Otherwise, the definition of $\partial\calA$ would imply that $O\cup T=C$ is fundable, contradicting $qm>qk\ge n$.
Since $F_{r+1}(\calA)\subseteq \partial\calA$, voter $i$ cannot add $T$ at step $r+1$.
Consider any nonempty ballot $S=A_j$ added by another voter.
By Lemma~\ref{lem:fundable-local}, there is a report $R$ such that every committee in $\calK_r(\calB)$ contains $S$, where $\calB=\calA[j\leftarrow R]$.
Lemma~\ref{lem:fundable-avoidance} gives $S\subseteq \partial\calB$.
Lemma~\ref{lem:fundable-replacement} then implies $\fund(\calA)\subseteq\fund(\calB)$, so $O$ remains fundable at $\calB$.
Voter $i$ still has ballot $T$, and $\calB$ still has at least $s$ empty ballots, since $A_j\ne\varnothing$.
The induction hypothesis under (ii) therefore gives
\[
  \calK_r(\calB)=\calK_r(\calB[i\leftarrow\varnothing]).
\]
Consequently, the same report $R$ ensures that $S$ is added at $\calA^-$ by Lemma~\ref{lem:fundable-local}.

In both cases, together with $F_r(\calA)=F_r(\calA^-)$, this proves $F_{r+1}(\calA)=F_{r+1}(\calA^-)$.
The equality of $\calK_{r+1}$ follows from $\calJ(\calA)=\calJ(\calA^-)$.
\end{proof}

The next proposition gives the main property of $\calK$: whenever a report admits a utility gain, the two families at the all-empty completion also admit a strict loss.

\begin{proposition}\label{prop:fundable-loss}
Fix a collection $D$ of $n-s-1$ known ballots and distinct ballots $T,R$.
Suppose that, for some completion $\bfU$, there exist $W\in\calK(D+T+\bfU)$ and $V\in\calK(D+R+\bfU)$ with $u_T(V)>u_T(W)$.
Then, there exist $W'\in\calK(D+T+\varnothing^s)$ and $V'\in\calK(D+R+\varnothing^s)$ with $u_T(V')<u_T(W')$.
\end{proposition}
\begin{proof}
Let $\calA=D+T+\varnothing^s$ and $\calA'=D+R+\varnothing^s$.
By Proposition~\ref{prop:fundable-families}, the committees $W,V$ in the assumption also belong to $\calK(\calA),\calK(\calA')$, respectively.
Define
\[
  t=|T|,\qquad g=\max(0,k-m+t),\qquad v=\max_{W\in\calK(\calA)}u_T(W).
\]
Here, $g$ is the minimum utility of $T$ over all size-$k$ committees.
Suppose that the claim is false.
Then, every $V\in\calK(\calA')$ satisfies $u_T(V)\ge v$.
The condition $u_T(V)>u_T(W)$ in the assumption implies $T\ne\varnothing$.
Lemma~\ref{lem:extension} gives a set containing $F(\calA)$ that satisfies JR and has at most $k-1$ candidates.
We can fill a remaining seat with a candidate in $T$ if needed, so $v\ge1$.

First, suppose that $g<v<t$.
Extend $F(\calA')$ to a fundable set $B\in\calJ_{\leq k}(\calA')$ with $|B|\le k-1$.
Every size-$k$ committee containing $B$ belongs to $\calK(\calA')$.
The minimum utility among these committees is
\[
  \max\bigl(|B\cap T|,\ k-|C\setminus T|\bigr)=\max(|B\cap T|,g),
\]
because the remaining seats can first be chosen from $C\setminus T$.
It follows that $|B\cap T|\ge v$.
The set $B\cap T$ remains fundable when $R$ is replaced by $T$, since every matching edge for a candidate in $B\cap T$ remains present.
By the definition of $\partial\calA$, we can add $F(\calA)$ and then apply Lemma~\ref{lem:extension} to obtain a fundable set $B'\in\calJ_{\leq k}(\calA)$ with $|B'|\le k-1$.
This set contains at least $v$ candidates in $T$.
If it contains exactly $v$, add a candidate in $T\setminus B'$, which exists since $v<t$.
We then fill the remaining seats to obtain a committee that belongs to $\calK(\calA)$ and has utility larger than $v$, a contradiction.

Next, suppose that $v=t$.
Every committee in $\calK(\calA')$ then contains $T$.
At the fixed point, Lemma~\ref{lem:fundable-local} gives $T\subseteq F(\calA)$.
Hence, every committee in $\calK(\calA)$ has utility $t$, making the condition $u_T(V)>u_T(W)$ in the assumption impossible.

Finally, suppose that $v=g\ge1$.
Every committee in $\calK(\calA)$ has utility exactly $g$ and therefore contains $O=C\setminus T$.
Lemma~\ref{lem:fundable-avoidance} gives $O\subseteq \partial\calA$, so $O$ is fundable.
Moreover, $g\ge1$ implies $t\ge m-k+1$.
By condition (ii) of Lemma~\ref{lem:fundable-empty-ballot} and Proposition~\ref{prop:fundable-families},
\[
  \calK(\calA')\subseteq\calK(D+\varnothing^{s+1})=\calK(\calA).
\]
Every committee in each family has utility $g$, again contradicting the assumption.
These cases cover all possible values because $g\le v\le t$ and $v\ge1$.
\end{proof}

It remains to choose one committee at every profile while retaining a strict-loss completion for every potentially profitable report.
Suppose $n>5$ and set $s=5$.
We use independent random choices over the committee families, together with completions that keep these families unchanged.

Let
\[
 a=\min(k-1,m-k),\qquad
 \calL=\{C\setminus S:S\subseteq C,\ |S|\le a\},\qquad
 L=|\calL|=\sum_{j=0}^a\binom mj,
\]
and put $\Sigma=\calL\cup\{\varnothing\}$.
Every ballot in $\calL$ has size at least $m-a=\max(k,m-k+1)$, so replacing it by an empty ballot does not change $\calK$, by condition (i) of Lemma~\ref{lem:fundable-empty-ballot}.
Since $a<m$, the empty ballot does not belong to $\calL$, and $|\Sigma|=L+1$.
Let $\calU_5$ be the set of multisets of $5$ ballots from $\Sigma$, and write
\[
 B=2^m,\qquad H=\binom mk,\qquad \xi=|\calU_5|=\binom{L+5}{5}.
\]
Here, $B$ counts ballot types, $H$ counts committees, and the last equality counts multisets of size $5$ from $L+1$ types.
For any $\bfU\in\calU_5$, replacing its ballots by empty ballots one at a time gives
\begin{equation}\label{eq:fundable-neutral}
 \calK(D+T+\bfU)=\calK(D+T+\varnothing^5),\qquad
 \calK(D+R+\bfU)=\calK(D+R+\varnothing^5).
\end{equation}

To obtain an anonymous rule, associate one random variable $X_{\calA}$ with each ballot multiset $\calA$ of size $n$.
Choose these variables independently, with $X_{\calA}$ uniform on $\calK(\calA)$.
We will fix an assignment for which no potentially profitable report is safe at every completion in $\calU_5$.

\begin{definition}\label{def:fundable-bad-events}
Let $D$ be a multiset of $n-6$ ballots, and let $T,R\subseteq C$ be distinct ballots.
Suppose that there are a completion $\bfU^*$ of five ballots from $2^C$ and committees $W,V$ such that
\[
 W\in\calK(D+T+\bfU^*),\qquad
 V\in\calK(D+R+\bfU^*),\qquad
 u_T(V)>u_T(W).
\]
The corresponding \emph{Bad Event} is
\[
 \sfB_{D,T,R}=
 \left\{u_T(X_{D+R+\bfU})\ge u_T(X_{D+T+\bfU})
 \text{ for every }\bfU\in\calU_5\right\}.
\]
The family of bad events is
\[
 \calE=\{\sfB_{D,T,R}:D,T,R\text{ satisfy the above conditions}\}.
\]
\end{definition}

The event indices depend only on the committee families, so they can be determined before choosing the values of $X_{\calA}$.
Any safe and profitable report with $n-6$ known ballots satisfies the index condition and makes the corresponding bad event occur.
Avoiding all events in $\calE$ therefore gives the required mechanism.

We next bound the probability of a strict loss at one completion.
The local-lemma estimate in Lemma~\ref{lem:fundable-constant} requires a sharper bound when $k=2$.

\begin{lemma}\label{lem:fundable-loss-probability}
Fix an event index $(D,T,R)$ and put $\calA_0=D+T+\varnothing^5$ and $\calB_0=D+R+\varnothing^5$.
Let $X,Y$ be independent uniform choices from $\calK(\calA_0),\calK(\calB_0)$, respectively.
Then,
\[
 \Pr(u_T(X)>u_T(Y))\ge\rho_{m,k},\qquad
 \rho_{m,k}=
 \begin{cases}
  1/(2H),&k=2,\\[2pt]
  4/(H+2)^2,&k\ge3.
 \end{cases}
\]
\end{lemma}
\begin{proof}
Proposition~\ref{prop:fundable-loss} gives a strict-loss pair in these two families.
Choose a noninteger threshold between the two utility values, and let $h$ of the $H$ committees have utility below that threshold.
If the truthful family has $r\ge1$ committees above the threshold, its probability of selecting one is at least $r/(h+r)\ge1/(h+1)$.
By the same argument, the reported family selects a committee below the threshold with probability at least $1/(H-h+1)$.
Independence gives a strict-loss probability of at least
\[
 \frac1{(h+1)(H-h+1)}\ge\frac4{(H+2)^2},
\]
which proves the result for $k\ge3$.

For $k=2$, consider a profile $\calA$ with an empty ballot and define
\[
 C_q(\calA)=\{c\in C:|\supp_{\calA}(\{c\})|\ge q\},\qquad
 \Gamma_{\calA}=(C,\calK(\calA)).
\]
The graph's vertices are candidates and its edges are committees.
Since fewer than $2q$ ballots are nonempty, the only nonempty fundable sets are the singletons $\{c\}$ with $c\in C_q(\calA)$.
If $C_q(\calA)=\varnothing$, every committee satisfies JR and $F(\calA)=\varnothing$, so the graph is complete.
Otherwise, each such singleton is maximal fundable, and Lemmas~\ref{lem:coloop} and~\ref{lem:extension} give $F(\calA)\subseteq\{c\}$ and JR for every committee containing $c$.
Thus, every $c\in C_q(\calA)$ is a \emph{Universal Vertex}, adjacent to all other vertices.

Changing the manipulator's ballot gives
\[
 C_q(\calA_0)\setminus T\subseteq C_q(\calB_0),\qquad
 C_q(\calB_0)\cap T\subseteq C_q(\calA_0).
\]
If either set $C_q(\calA_0),C_q(\calB_0)$ is empty or they intersect, the two graphs share a universal vertex.
Otherwise, these inclusions give a truthful universal vertex in $T$ and a reported universal vertex in $C\setminus T$.
Put $t=|T|$ and $o=m-t$; the strict-loss pair requires $t,o\ge1$.

First suppose there is a common universal vertex.
Exchanging $(X,Y,T)$ with $(Y,X,C\setminus T)$ preserves the strict-loss event, so we may assume that this vertex lies outside $T$.
For $o\ge2$, its incident edges include $t$ truthful edges of positive utility and $o-1$ reported edges of zero utility.
There are at most $\binom o2$ zero-utility committees and $to+\binom t2$ positive-utility committees, giving
\begin{equation}\label{eq:fundable-two-seat-product}
 \Pr(u_T(X)>u_T(Y))\ge
 \frac{t}{\binom o2+t}\,
 \frac{o-1}{o-1+to+\binom t2}
 \ge\frac1{m(m-1)}=\frac1{2H}.
\end{equation}
For the last inequality, bound the two denominators by
\[
 \begin{aligned}
  \binom o2+t&\le\frac{o(m-1)}2,\\
  o-1+to+\binom t2&\le\frac{t(t+4o-3)}2
  \le\frac{2t(o-1)m}{o}.
 \end{aligned}
\]
These bounds use $t\ge1$, $o\ge2$, and $t(3o-4)\ge o$.
If $o=1$, both graphs have exactly $t$ utility-one edges, and a strict loss requires a utility-two edge in the truthful graph.
The strict-loss probability is then at least $\frac1{t+1}\frac tH\ge1/(2H)$.

Finally, suppose there is no common universal vertex.
The same exchange preserves this case and lets us assume $o\ge2$.
The truthful universal vertex supplies $m-1$ positive-utility edges, and the reported one supplies $o-1$ zero-utility edges, so
\[
 \Pr(u_T(X)>u_T(Y))\ge
 \frac{m-1}{m-1+\binom o2}\,\frac{o-1}{H}
 \ge\frac{2(o-1)}{(o+1)H}\ge\frac2{3H}>\frac1{2H}.
\]
\end{proof}

By~\eqref{eq:fundable-neutral}, the probability bound in Lemma~\ref{lem:fundable-loss-probability} applies to every completion in $\calU_5$.
The associated strict-loss events need not be independent, because two completions may use a common profile variable.
We account for these dependencies to bound the probability that all $\xi$ completions fail to give a strict loss.

\begin{lemma}\label{lem:fundable-event-bound}
Every event in $\calE$ satisfies
\[
 \Pr(\sfB_{D,T,R})\le\exp(-\xi\rho_{m,k}).
\]
\end{lemma}
\begin{proof}
First suppose that at least one of $T,R$ lies outside $\Sigma$.
Two distinct completions give different profiles on the same side.
An equality $D+T+\bfU=D+R+\bfU'$ would imply $\bfU-\bfU'=R-T$, where subtraction uses the multiplicity of each ballot type.
The left-hand side uses only types in $\Sigma$, whereas the right-hand side has a nonzero coefficient outside $\Sigma$ because $T\ne R$.
Thus, all $2\xi$ profile variables are distinct, and independence gives
\[
 \Pr(\sfB_{D,T,R})\le(1-\rho_{m,k})^{\xi}.
\]

Now suppose $T,R\in\Sigma$.
Replacing the report and the ballots in $\bfU$ by empty ballots gives the same family $\calK_*:=\calK(D+\varnothing^6)$ at every relevant profile.
Condition (i) of Lemma~\ref{lem:fundable-empty-ballot} holds at every profile and applies also to the nonempty report $T$ or $R$.
Define the directed graph $G=(V_G,E_G)$ by
\[
 \begin{split}
 V_G&=\bigcup_{\bfU\in\calU_5}\{D+T+\bfU,D+R+\bfU\},\\
 E_G&=\{(D+T+\bfU,D+R+\bfU):\bfU\in\calU_5\}.
 \end{split}
\]
Each vertex corresponds to its profile variable, and each edge corresponds to one completion.
At a fixed vertex $\calA$, each equation $\calA=D+T+\bfU$ or $\calA=D+R+\bfU$ determines at most one $\bfU$.
Thus, every vertex has in-degree and out-degree at most one; distinct completions give distinct edges, so $|E_G|=\xi$.
Every edge corresponds to the replacement of one $T$ ballot by $R$; its target therefore has one less copy of $T$ than its source.
Hence, there is no directed cycle or loop.
An undirected cycle would have one incoming and one outgoing edge at every vertex and would therefore be directed.
Consequently, $G$ consists of disjoint directed paths.

The strict-loss pair gives at least two utility values in $\calK_*$.
Choose a noninteger threshold between two such values and label a vertex $1$ if its selected committee has utility above the threshold, and $0$ otherwise.
The labels are independent and have the same distribution.
If $N_*=|\calK_*|\le H$ and $j$ committees lie above the threshold, put
\[
 p=\frac{j}{N_*},\qquad
 \delta=p(1-p)=\frac{j(N_*-j)}{N_*^2}
 \ge\frac{N_*-1}{N_*^2}\ge\frac{H-1}{H^2}\ge\frac1{2H}\ge\rho_{m,k}.
\]
Here, $1\le j\le N_*-1$, $(x-1)/x^2$ is nonincreasing for $x\ge2$, and $(H+2)^2\ge8H$.

A label pair $1\to0$ on an edge gives a strict loss.
Thus, $\sfB_{D,T,R}$ requires every path to contain no such pair.
Let $P_r$ be the probability that a path with $r$ edges has no such pair.
Then,
\[
 P_0=1,\qquad P_1=1-\delta,\qquad
 P_r=P_{r-1}-\delta P_{r-2}\quad(r\ge2).
\]
Indeed, the first $r-1$ edges satisfy the condition and the last edge fails it exactly when the last two labels are $1,0$ and the first $r-2$ edges satisfy the condition.
The preceding edge has target label $1$ and always satisfies the condition, while the last two labels are independent of all earlier labels, giving the term $\delta P_{r-2}$.
Since removing an edge can only increase the probability, $P_{r-2}\ge P_{r-1}$ and therefore $P_r\le(1-\delta)P_{r-1}\le(1-\delta)^r$.
Different paths have disjoint vertex sets, so their conditions are independent.
Multiplying their bounds gives $\Pr(\sfB_{D,T,R})\le(1-\delta)^{\xi}\le(1-\rho_{m,k})^{\xi}$.
In both cases, $1-x\le\e^{-x}$ proves the required bound.
\end{proof}

To obtain a choice for which no bad event holds, we also bound the number of events that use a common profile variable.

\begin{definition}\label{def:fundable-dependency-graph}
The \emph{Dependency Graph} $G_{\mathrm{dep}}=(\calE,E_{\mathrm{dep}})$ has one vertex for each bad event in $\calE$.
Two distinct events are adjacent if and only if they use a common profile variable $X_{\calA}$.
\end{definition}

\begin{lemma}\label{lem:fundable-dependence}
The maximum degree $d$ of $G_{\mathrm{dep}}$ satisfies
\[
 d+1\le4B^2\xi^2.
\]
\end{lemma}
\begin{proof}
Each event uses at most $2\xi$ profile variables.
A fixed variable $X_{\calA}$ belongs to at most $2B^2\xi$ events: choose the truthful or reported side, the two reports, and a completion in $\calU_5$.
After these choices, removing the completion and the chosen report from $\calA$ gives a unique $D$.
For an upper bound, we can include choices for which these ballots cannot be removed from $\calA$.
Hence, at most $(2\xi)(2B^2\xi)=4B^2\xi^2$ events use a variable of a fixed event, including that event itself.
The bound does not depend on $n$, since profiles are ballot multisets.
\end{proof}

We use the symmetric form of the \emph{Lov\'asz Local Lemma}~\cite[Corollary~5.1.2]{alon2008probabilistic}.
If events have probability at most $p_0$, their dependency graph has maximum degree at most $d$, and $\e p_0(d+1)\le1$, then with positive probability no event occurs.
Each event is independent of the joint outcomes of its nonneighbors, since those events use only variables outside its own set of independent profile variables.
By Lemmas~\ref{lem:fundable-event-bound} and~\ref{lem:fundable-dependence}, it therefore suffices to prove the following inequality.

\begin{lemma}\label{lem:fundable-constant}
For every $2\le k<m$,
\[
 \xi\rho_{m,k}>1+\ln4+2m\ln2+2\ln \xi.
\]
\end{lemma}
\begin{proof}
We use $\ln2<7/10$.
For $k=2$, we have $a=1$, $L=m+1$, and $\xi=\binom{m+6}{5}$.
Write $f(m)=\binom{m+6}{5}/[m(m-1)]$ for the left-hand side of the required inequality.
For $m=3,4,5$, respectively, $\xi$ is $126,252,462$ and $f(m)$ is $21,21,231/10$.
Since these values of $\xi$ are less than $2^7,2^8,2^9$, the corresponding right-hand sides are less than $82/5,96/5,22$, proving the three initial cases.
Moreover,
\[
 f(m+2)-2f(m+1)+f(m)
 =\frac{(m+4)(m+5)(m+6)(m^2-5m+16)}{20m(m-1)(m+1)(m+2)}>0,
\]
because $m^2-5m+16=(m-5/2)^2+39/4$.
Thus, for $m\ge5$, the increase in $f$ is at least $f(6)-f(5)=33/10$.
The increase in the right-hand side is
\[
 2\ln2+2\ln\frac{m+7}{m+2}
 \le2\ln2+2\ln\frac{12}{7}<4\ln2<\frac{14}{5}<\frac{33}{10}.
\]
This proves the required inequality for every $m\ge3$ when $k=2$.

For $k\ge3$, define
\[
 g(x)=\frac{4\binom{x+4}{5}}{(x+1)^2}
 =\frac{x(x+2)(x+3)(x+4)}{30(x+1)}.
\]
First suppose $a=m-k$.
Then, $H\le L$ and $L+1\ge m+2\ge6$, so $\xi\rho_{m,k}\ge g(L+1)$.
For every integer $x\ge6$,
\[
 g(x+1)-g(x)
 =\frac{x^2}{10}+\frac{19x}{30}+\frac9{10}
  +\frac1{5(x+1)}-\frac1{5(x+2)}\ge\frac{83}{10},
\]
whereas
\[
 2\ln\frac{\binom{x+5}{5}}{\binom{x+4}{5}}
 =2\ln\left(1+\frac5x\right)\le\frac{10}{x}\le\frac53.
\]
Hence, the increase in $g(x)-2\ln\binom{x+4}{5}$ is larger than $2\ln2$.
At $x=6$, we have $g(6)=144/7>96/5$, while the right-hand side of the required inequality with $m=4$ and $\xi=\binom{10}{5}=252$ is less than $96/5$.
Since $(L+1)-6\ge m-4$, it follows that
\[
 \begin{split}
 g(L+1)-2\ln \xi
 &\ge g(6)-2\ln\binom{10}{5}+2(L-5)\ln2\\
 &>1+\ln4+2m\ln2.
 \end{split}
\]

Finally, suppose $k\ge3$ and $a=k-1$.
Then, $m\ge2k-1\ge5$ and
\[
 L+1\ge2+m+\binom m2>\frac{m^2}{2},\qquad
 H=\frac{m-k+1}{k}\binom m{k-1}\le\frac{m-2}{3}L.
\]
As $(m-2)/3\ge1$, we have $H+2\le(m-2)(L+2)/3$.
The formula for $g$ gives $g(x)>x^3/30$, and therefore
\[
 \xi\rho_{m,k}\ge\frac9{(m-2)^2}g(L+1)
 >\frac{3m^6}{80(m-2)^2}>12m.
\]
For the last inequality, put $m=x+5$ with $x\ge0$; then
\[
 m^5-320(m-2)^2
 =x^5+25x^4+250x^3+930x^2+1205x+245>0.
\]
On the other hand, $L+1\le2^m+1$ and $2^m+5\le2^{m+1}$ give $\xi\le(L+5)^5\le2^{5(m+1)}$.
The right-hand side of the required inequality is consequently at most
\[
 1+(12m+12)\ln2<\frac{42m+47}{5}<11m\qquad(m\ge5).
\]
Thus, the required inequality holds in all cases.
\end{proof}

To define $\mu$, fix a candidate order for GreedyAV and an order on all assignments of one committee to each ballot multiset of size $n$.
For $n>5$, compute the families $\calK(\calA)$ with $s=5$ and all event indices from Definition~\ref{def:fundable-bad-events}.
Enumerate the assignments with one committee from each family, and take the first assignment avoiding every bad event.
Write $W_{\calA}^*$ for the committee assigned to $\calA$.
Define
\[
 \mu(\calA)=
 \begin{cases}
 \GreedyAV(\calA),&n\le5,\\
 W_{\calA}^*,&n>5.
 \end{cases}
\]

\begin{theorem}\label{thm:jr-high-rat}
For every $n\ge1$ and $1\le k\le m$, the mechanism $\mu$ satisfies JR and $\rd(\mu)\ge n-5$.
It selects exactly $k$ candidates, is anonymous, and can be computed by a finite algorithm.
\end{theorem}
\begin{proof}
For $n\le5$, GreedyAV satisfies JR and $\rd(\mu)\ge0\ge n-5$.
Suppose $n>5$.
There are $\binom{n+B-1}{B-1}$ ballot multisets of size $n$.
By Proposition~\ref{prop:fundable-families}, a nonempty family $\calK(\calA)$ at each profile can be computed in a finite number of steps.
The event indices in Definition~\ref{def:fundable-bad-events} can also be computed by considering all $D,T,R$, all completions, and all pairs of committees in the corresponding families.
There are finitely many choices of these committees, and every index has $\xi$ specified completions to check.
Lemma~\ref{lem:fundable-constant} and the Lov\'asz Local Lemma give a choice for which no event holds, so the enumeration terminates.
Every chosen committee has size $k$ and satisfies JR, and the selection depends only on the ballot multiset.

A safe and profitable report with $n-6$ known ballots would satisfy the index condition and imply the corresponding event, contradicting the choice of committees.
If such a report is possible with less known information, take a profitable completion and add ballots from that completion to the known information so that exactly $n-6$ other ballots are known.
Safety holds for the remaining completions, and the chosen profitable completion remains possible.
This again contradicts the choice of committees, proving $\rd(\mu)\ge n-5$.
\end{proof}

\subsection{A Polynomial-Time Construction}\label{sec:jr-polynomial}

For a bounded number of candidates, we can replace the exhaustive search in the finite construction by a deterministic local-lemma algorithm.
For larger candidate sets, we use an algebraic selection rule whose output can be controlled by seven unknown ballots.
When $k\ge3$, two of these ballots encode committee certificates and the other five control the output.
When $k=2$, all certificates can be checked directly, so all seven ballots are available to control the output.
Combining the two constructions gives a polynomial-time JR mechanism with RAT-Degree at least $n-7$.

We retain the quota and notation from Sect.~\ref{sec:jr-finite}; in particular, making ballots empty does not change $q$.
Fix a candidate order $c_0,\ldots,c_{m-1}$.
Order sets of the same size by \emph{Colexicographic Order}, denoted by $\prec$: the larger of two sets contains the largest candidate belonging to exactly one of them.
The binary index of a ballot $U$ is $\sum_{c_i\in U}2^i$.
Order ballot types by these indices and ballot multisets lexicographically by their multiplicity vectors.
Use these orders for all subsequent choices.

We first replace the exhaustive search in $\mu$ by the deterministic algorithm of Moser and Tardos~\cite[Theorem~1.4]{moser2010constructive}.
Let $\nu$ agree with $\mu$ when $n\le5$.
Otherwise, compute the families $\calK(\calA)$ for all ballot multisets of size $n$ with $s=5$, and construct the variables $X_{\calA}$ and the event family $\calE$ from Definition~\ref{def:fundable-bad-events}.
Apply that algorithm to choose one value for every $X_{\calA}$ so that no event occurs, and let $\nu(\calA)$ be the chosen value at the input profile.
The fixed orders ensure that every input with the same $n,m,k$ uses the same global assignment.

\begin{corollary}\label{cor:bounded-candidates}
For every fixed $m_0$, the mechanism $\nu$ runs in polynomial time on profiles with $m\le m_0$, when the input is a list of the $n$ ballots.
It is anonymous, selects exactly $k$ candidates, satisfies JR, and has $\rd(\nu)\ge n-5$.
\end{corollary}
\begin{proof}
It suffices to consider $n>5$.
For fixed $m$, the number of ballot multisets of size $n$ is
\[
 P_n=\binom{n+B-1}{B-1}=n^{O(B)}.
\]
The recurrence~\eqref{eq:fundable-recurrence} can be computed on these multisets, and at most $mP_n$ iterations add a candidate to some set $F_r(\calA)$.
For each profile, it suffices to consider the manipulator's ballot type, the reported type, and the multiset of the five ballots made empty.
JR can be checked for each committee by counting the unrepresented approvers of each candidate.
There are only constantly many such choices and committees for fixed $m$, so all families $\calK$ can be computed in polynomial time.
The event indices can also be found in polynomial time by enumerating $D,T,R$, the five-ballot completions, and pairs of committees in the corresponding families.
There are at most $B^2P_{n-6}$ events.
Each event involves at most $2\xi$ variables, each with at most $H$ values.
A conditional event probability can therefore be computed exactly by enumerating at most $H^{2\xi}$ assignments to the unfixed variables.
This number is constant for fixed $m$.

We check the constant-degree and slack conditions in the deterministic local lemma.
Put $D_0=4B^2\xi^2$ and $p_0=\exp(-\xi\rho_{m,k})$.
By Lemma~\ref{lem:fundable-dependence}, the dependency graph has maximum degree at most $D_0-1$.
Lemma~\ref{lem:fundable-constant} gives $\e D_0p_0<1$.
With $x_0=1/D_0$ and $\varepsilon=1-\e D_0p_0>0$, every event $\sfB$ therefore satisfies
\[
 \Pr(\sfB)\le p_0=\frac{1-\varepsilon}{\e D_0}
 \le(1-\varepsilon)x_0(1-x_0)^{|\Gamma(\sfB)|},
\]
where $\Gamma(\sfB)$ is the set of neighbors of $\sfB$ in the graph $G_{\mathrm{dep}}$ from Definition~\ref{def:fundable-dependency-graph}.
Here, we used $(1-1/D_0)^{D_0-1}\ge\e^{-1}$.
Both $D_0$ and $\varepsilon$ depend only on $m,k$.
Over the finitely many pairs $2\le k<m\le m_0$, the degree bound $D_0$ is bounded by a constant and the slack $\varepsilon$ is bounded below by a positive constant.
Theorem~1.4 of Moser and Tardos~\cite{moser2010constructive} now gives a deterministic algorithm polynomial in the number of variables, events, and domain values, which is polynomial in $n$.
The fixed orders make the assignment independent of voter labels.
The proof of Theorem~\ref{thm:jr-high-rat} applies to any assignment avoiding all events, so it gives JR and the $n-5$ bound for $\nu$.
\end{proof}

For the algebraic construction, we need a committee family whose membership can be certified by matching computations.
We obtain such a family by allowing certain ballot changes and extending maximal fundable sets at the resulting profiles.
The allowed changes use the two conditions in Lemma~\ref{lem:fundable-empty-ballot}, with just one empty ballot required in condition (ii).

\begin{definition}\label{def:hash-range}
A \emph{Reachable Profile} from $\calA$ is obtained by a finite sequence of the following ballot changes.
An empty ballot may be replaced by any ballot.
A ballot $U$ may be replaced by an empty ballot if either condition holds:
\begin{enumerate}
\item[(i)] $U\in\calL$;
\item[(ii)] an empty ballot exists, $|U|\ge m-k+1$, and $C\setminus U$ is fundable at the profile before the step.
\end{enumerate}
Zero steps are allowed.

Define
\[
 \calR(\calA)=
 \left\{W\in\calW_k:
 \begin{array}{l}
  \exists\,\calB\text{ reachable from }\calA,\ \exists\,Q\subseteq W,\\
  Q\in\fund(\calB),\\
  \forall c\in C\setminus Q,\quad Q\cup\{c\}\notin\fund(\calB)
 \end{array}
 \right\}.
\]
The last condition means that $Q$ is inclusion-maximal in $\fund(\calB)$.
\end{definition}

Every allowed change preserves existing JR violations for size-$k$ committees.
A ballot that can be made empty is represented by every such committee, and an empty ballot cannot belong to a cohesive group.
Thus, a committee satisfying JR at a reachable profile also satisfies JR at the original profile.
By Lemma~\ref{lem:extension}, maximal fundable sets and their size-$k$ extensions satisfy JR, giving $\calR(\calA)\subseteq\calJ(\calA)$.
Concatenating two sequences of allowed changes gives
\begin{equation}\label{eq:hash-reach}
 \calB\text{ reachable from }\calA
 \quad\Longrightarrow\quad
 \calR(\calB)\subseteq\calR(\calA).
\end{equation}

The certificate test will depend on the ballots that remain fixed and the number of positions that can change.
We first identify those positions.
Let $J_0(\calA)$ consist of the positions of the empty ballots and the ballots in $\calL$ at $\calA$.
If $J_0(\calA)=\varnothing$, put $J(\calA)=\varnothing$.
Otherwise, form an auxiliary profile $\widehat\calA$ by replacing the ballots in $J_0(\calA)$ by full ballots except for one position that remains empty, and let
\[
 J(\calA)=J_0(\calA)\cup\{i\in N\setminus J_0(\calA):|A_i|\ge m-k+1,\ C\setminus A_i\in\fund(\widehat\calA)\}.
\]
The choice of the empty position gives the same fundability family, since each choice gives the same numbers of each ballot.
Write $\eta(\calA)=|J(\calA)|$ and $\calF(\calA)=(A_i)_{i\in N\setminus J(\calA)}$.
For $Q\subseteq C$, let $r_{\calF(\calA)}(Q)$ be the size of a maximum matching between the voters in $\calF(\calA)$ and $q$ copies of each candidate in $Q$, using the approval edges from Lemma~\ref{lem:hall}.
The value $q|Q|-r_{\calF(\calA)}(Q)$ is the number of additional voters needed to make $Q$ fundable.

\begin{proposition}\label{prop:hash-certificates}
The profiles reachable from $\calA$ are exactly those with arbitrary ballots in $J(\calA)$ and $\calF(\calA)$ fixed.
A committee $W$ belongs to $\calR(\calA)$ if and only if there exists $Q\subseteq W$ such that
\[
 \begin{split}
 |W|=k,\qquad q|Q|-r_{\calF(\calA)}(Q)&\le \eta(\calA),\\
 r_{\calF(\calA)}(Q\cup\{c\})-r_{\calF(\calA)}(Q)&<q\qquad(c\in C\setminus Q).
 \end{split}
\]
The set $J(\calA)$ can be computed and a given certificate $Q$ can be verified using $O(n+m)$ maximum-matching computations.
\end{proposition}
\begin{proof}
We first show why testing each position once suffices to identify $J(\calA)$.
If $C\setminus U$ is fundable in the auxiliary profile, replacing $U$ by $C$ preserves the entire fundability family.
Indeed, the support of any set with a nonempty intersection with $U$ is unchanged, while every set disjoint from $U$ lies in its already fundable complement.
By Lemma~\ref{lem:hall}, every set that was not fundable remains so.
Thus, each additional selected ballot can be replaced first by an empty ballot and then by $C$, without changing the conditions for other ballots.
One position remains empty throughout these steps, as required by condition (ii).
After these steps, all ballots in $J(\calA)$ are empty or full, so they can be replaced by arbitrary ballots.

Conversely, consider the first change to a position outside $J(\calA)$ in any allowed chain of steps.
Its ballot is nonempty and outside $\calL$, so the change must use condition (ii).
All earlier modified positions lie in $J(\calA)$, and at least one of them is empty.
Replacing the other ballots in $J(\calA)$ by $C$ can only add matching edges.
By the argument above, the resulting profile has the same fundability family as $\widehat\calA$, since the particular empty position in $J(\calA)$ does not change fundability.
The required complement was therefore already fundable in $\widehat\calA$, contradicting the assumption that its position was not selected.
If $J_0(\calA)=\varnothing$, no first step is possible.

It remains to characterize the sets that can be maximal fundable at a profile with arbitrary ballots in $J(\calA)$.
Suppose $Q$ is maximal at such a profile, and consider a matching that makes $Q$ fundable.
For each voter in $J(\calA)$, keep only the approval corresponding to her matching edge, or make her ballot empty if she is not matched.
This preserves fundability and maximality.
The matching of voters in $\calF(\calA)$ into the copies of $Q$ can be augmented to a maximum matching without making any matched copy unused.
At each step, one additional copy becomes matched to a voter in $\calF(\calA)$; make the ballot of its old singleton voter from $J(\calA)$ empty.
The resulting profile makes $Q$ fundable using exactly $q|Q|-r_{\calF(\calA)}(Q)\le \eta(\calA)$ singleton voters in $J(\calA)$ and keeps $Q$ maximal.

For any maximum matching from $\calF(\calA)$ into the copies of $Q$, add singleton voters for the unused copies in this way.
At the resulting profile, a candidate $c\notin Q$ can be added exactly when $r_{\calF(\calA)}(Q\cup\{c\})=r_{\calF(\calA)}(Q)+q$.
One direction follows because the singleton voters give only $q|Q|-r_{\calF(\calA)}(Q)$ matching edges.
For the other direction, the fixed-voter matching can be augmented in the graph that also contains the $q$ copies of $c$, with all already matched copies remaining matched.
An additional copy of $Q$ cannot become matched, since that would give a fixed-voter matching into $Q$ of size larger than $r_{\calF(\calA)}(Q)$.
If $r_{\calF(\calA)}(Q\cup\{c\})=r_{\calF(\calA)}(Q)+q$, all copies of $c$ become matched while all already matched copies of $Q$ remain matched.
The singleton voters are still matched to the other copies of $Q$.
Thus, maximality implies the certificate inequalities.
Conversely, if those inequalities hold, use $q|Q|-r_{\calF(\calA)}(Q)$ singleton voters to complete a maximum matching into $Q$ and leave the other positions of $J(\calA)$ empty.
No candidate can be added, so $Q$ is maximal at this reachable profile.
Preprocessing checks at most $n$ complement sets, and the verification of the certificate conditions requires at most $m+1$ matching computations.
\end{proof}

We also define a fallback committee $W^{\mathrm{fb}}(\calA)$ for profiles at which the proposed committee has no certificate found by our test.
Start with $Q=\varnothing$ and add candidates in the fixed order whenever the resulting set is fundable using only $\calF(\calA)$.
Extend the resulting maximal fundable set $Q$ to size $k$ in the same order.
The set $Q$ certifies this committee in Proposition~\ref{prop:hash-certificates}: we have $r_{\calF(\calA)}(Q)=q|Q|$, and no candidate can be added using only $\calF(\calA)$.

We next show that $\calR$ has the strict-loss property established for $\calK$ in Proposition~\ref{prop:fundable-loss}.
The intermediate-utility case uses the same matching argument.
For the two extreme cases, reachability supplies the candidate requirements previously enforced by $F$.

\begin{lemma}\label{lem:hash-reverse}
Fix $D,T,R$ and $s\ge1$ as in Proposition~\ref{prop:fundable-loss}.
If some completion $\bfU$ admits $W\in\calR(D+T+\bfU)$ and $V\in\calR(D+R+\bfU)$ with $u_T(V)>u_T(W)$, then there exist $W'\in\calR(D+T+\varnothing^s)$ and $V'\in\calR(D+R+\varnothing^s)$ with $u_T(W')>u_T(V')$.
\end{lemma}
\begin{proof}
Put $\calA=D+T+\varnothing^s$ and $\calB=D+R+\varnothing^s$.
By~\eqref{eq:hash-reach}, the given committees belong to $\calR(\calA)$ and $\calR(\calB)$, respectively, and still give a utility gain under $R$.
Put $t=|T|$, $g=\max(0,k-m+t)$, and $v=\max_{W\in\calR(\calA)}u_T(W)$.
Suppose that no pair gives a strict loss under $R$.
Then, every committee in $\calR(\calB)$ has utility at least $v$.
As in Proposition~\ref{prop:fundable-loss}, the empty ballot leaves a seat available for a candidate in $T$, giving $1\le v\le t$.

If $g<v<t$, choose a maximal fundable set $Q$ at $\calB$.
Every size-$k$ extension of $Q$ belongs to $\calR(\calB)$, so its minimum possible utility is $\max(|Q\cap T|,g)\ge v$.
Since $g<v$, we have $|Q\cap T|\ge v$.
Replacing $R$ by $T$ preserves fundability of $Q\cap T$.
Extending this set as in Proposition~\ref{prop:fundable-loss} gives a committee in $\calR(\calA)$ with utility larger than $v$, a contradiction.

For the remaining cases, Lemma~\ref{lem:coloop} implies that $\partial\calA$ is also the intersection of all size-$k$ committees containing a maximal fundable set.
Indeed, since $m>k$, a maximal fundable set without $c$ can be extended to size $k$ without $c$.

Suppose first that $v=t$, so every committee in $\calR(\calB)$ contains $T$.
The intersection equality gives $T\subseteq\partial\calB'$ at every profile $\calB'$ reachable from $\calB$.
At $\calB$, an empty ballot exists, so $|T|\le k-1$.
Consider any allowed chain of steps from $\calA$, and apply the same steps starting at $\calB$, with the reports $T$ and $R$ fixed.
At corresponding profiles $\calA',\calB'$, Lemma~\ref{lem:fundable-replacement} gives $T\subseteq\partial\calA'$ and $\fund(\calA')\subseteq\fund(\calB')$.
Every allowed change at another voter's position is allowed on both sides.
A complement fundable at $\calA'$ remains fundable at $\calB'$, and any empty ballot at $\calA'$ is outside the manipulator's nonempty ballot $T$ and remains empty at $\calB'$.
The truthful ballot $T$ cannot be made empty.
Condition (i) fails because $|T|<k$ implies $T\notin\calL$.
Condition (ii) would make $C\setminus T$ fundable; together with $T\subseteq\partial\calA'$, this would make $C$ fundable, contradicting $qm>n$.
Thus, every truthful endpoint has all members of $T$ in every maximal fundable set.
All committees in $\calR(\calA)$ then give utility $t$, contradicting the assumed utility gain under $R$.

Finally, if $v=g$, every truthful committee has utility $g\ge1$ and contains $C\setminus T$.
The intersection equality and Lemma~\ref{lem:coloop} make $C\setminus T$ fundable at $\calA$.
Also, $g\ge1$ implies $|T|\ge m-k+1$.
Condition (ii) of Definition~\ref{def:hash-range} allows us to make $T$ empty and then replace it by $R$.
By~\eqref{eq:hash-reach}, $\calR(\calB)\subseteq\calR(\calA)$, so every reported committee also gives utility $g$.
This again contradicts the assumed utility gain.
\end{proof}

For $k=2$, a proposed committee has only four subsets to test as certificates.
For general $k$, checking all subsets requires $2^k$ certificate tests.
We therefore store most of each certificate as the complement of one unknown ballot.
The next lemma shows that at most four candidates need to be omitted from this encoding, so the certificate can be recovered by testing those missing candidates.
The remaining unknown ballots will lie in $\Sigma$; changing these ballots gives mutually reachable profiles and preserves the certificates.

\begin{lemma}\label{lem:two-certificate-ballots}
Let $\calA_0=D+T+\varnothing^s$ and $\calB_0=D+R+\varnothing^s$, where $s\ge3$.
Given certificates $Q_1\subseteq W$ and $Q_2\subseteq V$ at these two profiles, there exist sets $Q_1^0\subseteq Q_1$ and $Q_2^0\subseteq Q_2$ satisfying
\[
 |Q_j\setminus Q_j^0|\le4,\qquad |Q_j^0|\le k-1\qquad(j=1,2)
\]
with the following property.
If two unknown ballots are $C\setminus Q_1^0$ and $C\setminus Q_2^0$ and all remaining unknown ballots belong to $\Sigma$, then the resulting profiles $\calA,\calB$ satisfy $\calR(\calA)=\calR(\calA_0)$ and $\calR(\calB)=\calR(\calB_0)$.
Both original certificates remain valid in Proposition~\ref{prop:hash-certificates}.
\end{lemma}
\begin{proof}
We first find a common auxiliary profile at which most of each certificate is fundable.
Reserve the two certificate positions and a third unknown position $z$.
For each base profile, form the initial auxiliary profile in Proposition~\ref{prop:hash-certificates}, with position $z$ empty, and then make the manipulator's position and both certificate positions empty.
The resulting profile $\calH$ is the same on both sides: the profiles can differ only in the manipulator's ballot before that position is made empty.

Replace every ballot in $J(\calA_0)$ by $C$.
The inequality $q|Q_1|-r_{\calF(\calA_0)}(Q_1)\le\eta(\calA_0)$ makes $Q_1$ fundable at this profile.
Making position $z$ empty removes at most one matching edge, so we can remove at most one candidate to obtain a matching for the remaining candidates.
By the proof of Proposition~\ref{prop:hash-certificates}, replacing the additional selected ballots by full ballots does not change the fundability family of the initial auxiliary profile.
Use that auxiliary profile again and make the manipulator's position and the two certificate positions empty.
Remove at most one candidate for each of these three positions to obtain a set $Q_1^0$ fundable at $\calH$, with $|Q_1\setminus Q_1^0|\le4$.
Since $\calH$ has an empty ballot, every fundable set has at most $k-1$ candidates.
The same argument using $\calF(\calB_0)$ and $\eta(\calB_0)$ gives $Q_2^0$.

The resulting profiles are reachable from their base profiles because every unknown position was empty.
For the other direction, first make position $z$ empty; its ballot belongs to $\Sigma$, so this step does not require another empty position.
Replace the other original positions in $J_0$ and the remaining unknown positions by full ballots as needed.
Outside the manipulator's position and the two certificate positions, the resulting profile contains every approval in $\calH$.
Those three positions are empty in $\calH$, so both $Q_1^0$ and $Q_2^0$ remain fundable.
Each certificate ballot has size at least $m-k+1$, its complement is fundable, and position $z$ is empty.
Thus, by condition (ii) of Definition~\ref{def:hash-range}, we can make the first certificate ballot empty and then full.
These steps only add approvals, so the same steps are allowed for the second certificate ballot.
Finally, make these full ballots empty and then replace them by the original ballots, leaving all unknown positions empty.
Each resulting profile and its base profile are mutually reachable.
They therefore have the same reachable profiles and the same sets that can be maximal fundable at a reachable profile.
Proposition~\ref{prop:hash-certificates} gives the required equality of committee families and implies that both certificates still satisfy the certificate conditions.
\end{proof}

We use the following certificate test for a proposed committee $W$.
First, test every $Q\subseteq W$ with $|Q|\le2$.
Next, for each ballot $U$ with $Q^0=C\setminus U\subseteq W$ and $|Q^0|\le k-1$, test every set
\begin{equation}\label{eq:recover-certificate}
 Q=Q^0\cup Z,\qquad Z\subseteq W\setminus Q^0,\qquad |Z|\le4.
\end{equation}
Accept $W$ if one of these sets satisfies the certificate conditions in Proposition~\ref{prop:hash-certificates}.
This requires $O(m^2+nm^4)$ certificate tests.
For $k=2$, the first step checks every possible certificate; for $k\ge3$, the second step finds the certificates supplied by Lemma~\ref{lem:two-certificate-ballots}.

The algebraic rule $\sfS_{m,k}$ uses the ballots to propose one committee and accepts that committee only if it passes the certificate test.
To prescribe the proposed committee, we need enough distinct ballots in $\calL\setminus\{C\}$ to adjust four field coordinates.
We choose a field large enough to represent all $H$ committees and enough polynomial coefficients to encode any of the $2^m$ ballot types.
Put
\[
 \begin{gathered}
 M=L-1,\qquad r=\lceil\log_2 H\rceil,\qquad d=\lceil m/r\rceil,\\
 t=\begin{cases}3,&k=2,\\2,&k\ge3.\end{cases}
 \end{gathered}
\]
Each coordinate adjustment will use at most $t$ control ballots.
Let $b_1,\ldots,b_t$ be nonnegative integers with sum $r$ that differ by at most one, with the larger values first.
These are the sizes of $t$ consecutive blocks of binary coordinates; put $S_t=\sum_{j=1}^t(2^{b_j}-1)$.
The encoding requires at most $4S_t+d$ control ballots, so it is available whenever
\begin{equation}\label{eq:binary-capacity}
 4S_t+d\le M.
\end{equation}
When~\eqref{eq:binary-capacity} holds, use the encoding from Appendix~\ref{app:polynomial-selectors} over the field $\FF=\FF_{2^r}$ with a fixed binary basis.
Each ballot contributes four field elements and a polynomial of degree at most $d$; write
\[
 g(U)=(\alpha_U,\beta_U,\xi_U,\eta_U;P_U),\qquad
 (u,v,x,y;P)=\sum_i g(A_i).
\]
The selection formula is
\begin{equation}\label{eq:four-selector}
 \Psi(u,v,x,y;P)=ux+vy+vP(u).
\end{equation}
Let $z\in\{0,\ldots,2^r-1\}$ be the integer represented by this field value in the fixed basis, and let $W_0,\ldots,W_{H-1}$ be the size-$k$ committees in colexicographic order.
Define
\[
 \sfS_{m,k}(\calA)=
 \begin{cases}
 W_z,&z<H\text{ and }W_z\text{ passes the certificate test},\\
 W^{\mathrm{fb}}(\calA),&\text{otherwise}.
 \end{cases}
\]
Appendix~\ref{app:polynomial-selectors} proves the following property of this rule.
It also shows that~\eqref{eq:binary-capacity} holds for every $m\ge4096$ (Lemma~\ref{lem:selector-parameters}), leaving only a bounded number of candidate-set sizes for Corollary~\ref{cor:bounded-candidates}.

\begin{lemma}\label{lem:seven-selector}
For $2\le k<m$ satisfying~\eqref{eq:binary-capacity}, the procedure $\sfS_{m,k}$ is anonymous, runs in polynomial time, and selects a committee in $\calR(\calA)$.
Let $D$ consist of $n-8$ ballots and let $T,R\subseteq C$ be distinct ballots.
Suppose that $W\in\calR(D+T+\varnothing^7)$ and $V\in\calR(D+R+\varnothing^7)$ satisfy $u_T(W)>u_T(V)$.
Then, some completion $\bfU$ of $7$ ballots satisfies
\[
 u_T\bigl(\sfS_{m,k}(D+T+\bfU)\bigr)
 >u_T\bigl(\sfS_{m,k}(D+R+\bfU)\bigr).
\]
\end{lemma}

The complete mechanism is
\begin{equation}\label{eq:seven-mechanism}
 \mu_{\mathrm H}(\calA)=
 \begin{cases}
 \GreedyAV(\calA),&n\le7,\\
 \nu(\calA),&n>7\text{ and }4S_t+d>M,\\
 \sfS_{m,k}(\calA),&n>7\text{ and }4S_t+d\le M.
 \end{cases}
\end{equation}

\begin{theorem}\label{thm:jr-polynomial-rat}
For every $n\ge1$ and $1\le k\le m$, the mechanism $\mu_{\mathrm H}$ is anonymous, runs in polynomial time when the input is a list of the $n$ ballots, selects exactly $k$ candidates, satisfies JR, and has $\rd(\mu_{\mathrm H})\ge n-7$.
\end{theorem}
\begin{proof}
For $n\le7$, GreedyAV satisfies JR and the required RAT-Degree bound follows from $\rd\ge0$.
Assume that $n>7$.
If~\eqref{eq:binary-capacity} fails, Lemma~\ref{lem:selector-parameters} gives $m<4096$, and Corollary~\ref{cor:bounded-candidates} with $m_0=4095$ gives all the required properties, with the stronger bound $n-5$.
Otherwise, Lemma~\ref{lem:seven-selector} gives anonymity, polynomial time, and an output in $\calR(\calA)\subseteq\calJ(\calA)$.

To prove the RAT-Degree bound in this remaining case, fix $n-8$ known ballots $D$, a true ballot $T$, and a report $R\ne T$ that is profitable at some completion.
Both outputs belong to $\calR$, so Lemma~\ref{lem:hash-reverse} gives a strict-loss pair at the all-empty completion.
Lemma~\ref{lem:seven-selector} then gives a completion of the same seven unknown ballots at which the report gives strictly lower utility.
Hence, no safe and profitable report exists with $n-8$ known ballots.
If such a report was possible with at most $n-9$ known ballots, add ballots from a profitable completion to the known ballots so that exactly $n-8$ ballots are known.
Safety and that profitable completion would remain, a contradiction.
It follows that $\rd(\mu_{\mathrm H})\ge n-7$.
\end{proof}

We do not optimize the running-time bound.
The bounded-candidate branch computes an assignment over all ballot multisets, so its polynomial bound can have a very large exponent.
There is room for algorithmic improvement; we use this construction to keep the proof simple and do not pursue more involved algorithms to reduce the running time.

\section{RAT-Degree with Immutable Voter IDs}\label{sec:immutable-ids}

The following construction uses additional private credentials associated with the voters.
This kind of auxiliary information is common in election systems: a ballot may carry an ID, a voter may be asked to provide an identification number, or an online voting system may record network metadata.
We consider credentials that voters cannot choose or change, and assume that an unknown voter's credential and ballot are both unobserved.
Under this information model, we construct a polynomial-time JR mechanism with RAT-Degree at least $n-1$.

\begin{definition}\label{def:id-information}
Each voter $i$ has an \emph{Immutable Identifier} (ID) $\iota_i\in\ZZ_{>0}$, fixed before reporting and unavailable for strategic choice or alteration.
The IDs are pairwise distinct, and the mechanism observes the entire profile $((\iota_i,A_i))_{i\in N}$.
A manipulator knows her own ID and the pairs $(\iota_j,A_j)$ for a set $K\subseteq N\setminus\{i\}$.
For every remaining voter, both the ID and the ballot are unknown.
Admissible completions range over all ballots and all positive-integer IDs distinct from one another and from the known IDs, including the manipulator's ID.
In particular, there is no public list of the assigned IDs.

Safety and profitability are defined as in Definition~\ref{def:safe-profitable}, with completions now specifying both IDs and ballots.
Every comparison between truth and a report holds the entire ID vector $\boldsymbol\iota=(\iota_1,\ldots,\iota_n)$ fixed.
Write $\rd_{\mathrm{ID}}$ for the minimum $|K|$, over all admissible ID assignments and ballot profiles, for which a safe and profitable report exists, with value $n$ if none exists.
\end{definition}

We consider $2\le k<m$ and use the committee family $\calR$ from Sect.~\ref{sec:jr-polynomial}.
Lemma~\ref{lem:hash-reverse} gives a strict-loss pair whenever a report is profitable at some completion.
To select both committees of this pair, we let an ID encode a table with two entries and check the committee certificates before accepting its outputs.
Let $b(U)$ be the binary index of a ballot $U$ under the candidate order from Sect.~\ref{sec:jr-polynomial}, and put
\[
 \sigma(\calA)=\sum_{i\in N}b(A_i).
\]
Changing one ballot from $T$ to $R\ne T$ changes this integer by $b(R)-b(T)\ne0$.
Thus, two table entries indexed by $\sigma$ suffice to distinguish the two profiles.

Fix an encoding of ordered tables with at most two entries $(s,W,Q)$, where $s$ is a nonnegative integer and $W,Q\subseteq C$, using binary strings with length prefixes.
Each table must have encodings larger than any given set of known IDs.
For this purpose, use the unique representation
\[
 \iota=2^\ell(2e+1),\qquad \ell,e\in\ZZ_{\ge0}.
\]
Let $\mathcal T(\iota)$ be the table encoded by $e$, interpreting an invalid code as an empty table.
For each table code $e$, increasing $\ell$ gives arbitrarily large IDs encoding that same table.

Let $\iota^*=\max_i\iota_i$.
An entry $(s,W,Q)$ in $\mathcal T(\iota^*)$ is applicable at $\calA$ if $s=\sigma(\calA)$ and $Q$ certifies $W\in\calR(\calA)$ under Proposition~\ref{prop:hash-certificates}.
Using the fallback committee $W^{\mathrm{fb}}$ from Sect.~\ref{sec:jr-polynomial}, define
\begin{equation}\label{eq:id-selector}
 \mu_{\mathrm{ID}}(\calA,\boldsymbol\iota)=
 \begin{cases}
 W,&\text{for the first applicable entry }(s,W,Q)\text{ in }\mathcal T(\iota^*),\\
 W^{\mathrm{fb}}(\calA),&\text{if no entry applies}.
 \end{cases}
\end{equation}

\begin{theorem}\label{thm:id-rat}
For every $n\ge1$ and $2\le k<m$, the mechanism $\mu_{\mathrm{ID}}$ selects exactly $k$ candidates, satisfies JR, and has
\[
 \rd_{\mathrm{ID}}(\mu_{\mathrm{ID}})\ge n-1.
\]
It runs in polynomial time in the total input length, including the binary encodings of the IDs.
If $k\ge3$ and $k$ divides $n$, then $\rd_{\mathrm{ID}}(\mu_{\mathrm{ID}})=n-1$.
\end{theorem}
\begin{proof}
Every output belongs to $\calR(\calA)\subseteq\calJ(\calA)$.
Computing $\sigma$ and decoding the largest ID take polynomial time in the input length.
At most two certificates are checked, and both these checks and the fallback computation take polynomial time by Proposition~\ref{prop:hash-certificates}.

For the lower bound, suppose $n\ge2$ and fix a manipulator with true ballot $T$, a report $R\ne T$, and $n-2$ known pairs whose ballots form $D$.
Assume that the report is profitable at some completion of the remaining voter $j$.
Both outcomes belong to $\calR$, so Lemma~\ref{lem:hash-reverse} with $s=1$ gives committees $W,V$ such that
\[
 \calA_0=D+T+\varnothing,\qquad \calB_0=D+R+\varnothing,\qquad
 W\in\calR(\calA_0),\quad V\in\calR(\calB_0),\quad u_T(W)>u_T(V).
\]
Choose certificates $Q_W,Q_V$ by Proposition~\ref{prop:hash-certificates}.
The integers $\sigma(\calA_0)$ and $\sigma(\calB_0)$ are distinct.
Let $e$ encode the table
\[
 \bigl((\sigma(\calA_0),W,Q_W),\ (\sigma(\calB_0),V,Q_V)\bigr).
\]
Choose $\ell$ so that $\iota_j=2^\ell(2e+1)$ exceeds every known ID, including the manipulator's, and complete voter $j$'s ballot by $\varnothing$.
This is an admissible completion, and the same ID selects the table under both reports.
Its respective entries apply at $\calA_0$ and $\calB_0$, so the outcomes are $W$ and $V$.
The report therefore incurs a strict loss and is not safe.

A safe and profitable report with fewer known pairs could be extended to $n-2$ known pairs by fixing additional pairs from a profitable completion.
Hence, $\rd_{\mathrm{ID}}(\mu_{\mathrm{ID}})\ge n-1$; for $n=1$, this lower bound is immediate.

For the equality statement, fix any ID assignment and regard $\mu_{\mathrm{ID}}$ as a deterministic rule on approval ballots.
This rule satisfies JR.
When $k\ge3$ and $k$ divides $n$, Corollary~\ref{cor:scaled-impossibility} with $f(t)=t$ implies that it admits a profitable report at some profile.
Knowing all $n-1$ other pairs leaves only that completion, so the report is also safe.
Thus, $\rd_{\mathrm{ID}}(\mu_{\mathrm{ID}})\le n-1$, proving equality.
\end{proof}

The hidden ID varies across admissible completions and remains fixed within each comparison of reports.
The unrestricted positive-integer domain in Definition~\ref{def:id-information} ensures that one hidden ID can encode any two-entry table and exceed all known IDs.

\section{Future Work}

The upper bounds in Sect.~\ref{sec:standard-rat} do not determine the exact RAT-Degrees of GreedyAV, PAV, and MES for all parameter values.
Matching lower bounds would determine how much information is necessary for safe manipulation of these rules.

For the constructions in Sect.~\ref{sec:main}, the main question is whether a mechanism running in polynomial time in both $n$ and $m$ can achieve RAT-Degree at least $n-5$.
Corollary~\ref{cor:bounded-candidates} achieves this bound when $m$ is bounded by a constant.
The $n-7$ construction also leaves room for simpler mechanisms and smaller polynomial running-time bounds.
The optimal additive constant and the existence of analogous guarantees under stronger representation axioms remain open.

Section~\ref{sec:immutable-ids} achieves RAT-Degree at least $n-1$ when unknown voters' immutable IDs are also unobserved.
It remains to determine which guarantees are possible when IDs are public or their admissible values are restricted.

\bibliographystyle{alphaurl}
\bibliography{ref}

\appendix
\section{Polynomial Committee Selection}\label{app:polynomial-selectors}

We construct the encoding $g$ used by $\sfS_{m,k}$ for $2\le k<m$ satisfying~\eqref{eq:binary-capacity}.
The encoding will allow a common completion to select both committees of a strict-loss pair, as required by Lemma~\ref{lem:seven-selector}.

Use the field $\FF$ and the parameters $M,r,d,t,b_1,\ldots,b_t,S_t$ from Sect.~\ref{sec:jr-polynomial}.
Shoup's deterministic construction of irreducible polynomials in fixed characteristic~\cite[Sect.~3]{shoup1990}, applied in characteristic two with extension degree $r$, provides such a representation in time polynomial in $r$.
A field element's binary index is the integer represented by its basis coefficients.
Since $H\ge m\ge3$ and $r\ge2$, we have $d=\lceil m/r\rceil<m\le2^r$.
Partition the basis coordinates into consecutive blocks of sizes $b_1,\ldots,b_t$, beginning with the least significant coordinates.
Let $E_t$ consist of the nonzero field elements supported on a single block.
Then, $|E_t|=S_t$, and every field element is the sum of at most $t$ elements of $E_t$.

For each of the four coordinates, take the vectors with a single nonzero entry in that coordinate, one for each element of $E_t$.
For the first coordinate, also take the vectors with nonzero entry of binary index $1,\ldots,d$, counting each vector only once.
There are at most $4S_t+d\le M$ such vectors.
Order the vectors by coordinate and then by binary index.
Order $\calL\setminus\{C\}$ by complement size and then by the colexicographic order of the complement, and assign the vectors to the first ballots in this order.
These are the \emph{Control Ballots}; each contributes its assigned vector and the zero polynomial.
Since control ballots belong to $\Sigma$, replacing them by empty ballots or other control ballots preserves $\calR$ and its certificates.
Give $\varnothing$ and $C$ zero contribution.
Every remaining ballot $U$ has zero vector contribution and a monic polynomial $P_U$ of degree $d$.
Write the binary index of $U$ in base $2^r$, padding it with zeros to obtain $d$ digits.
Use these digits, from least to most significant, as the coefficients of $1,Z,\ldots,Z^{d-1}$, interpreting each digit as the binary index of a field element.
Since $(2^r)^d\ge2^m$, this assignment gives distinct polynomials to different ballots.
Thus, the only different ballots with equal contributions are $\varnothing$ and $C$.

Control ballots must prescribe both outputs using the same completion.
The next lemma shows that $2t+1$ such ballots suffice.
\begin{lemma}\label{lem:four-selector}
For two reports $T,R$ with different contributions, any common initial state, and any distinct $w_0,w_1\in\FF$, at most $2t+1$ control ballots make the truthful and reported values of~\eqref{eq:four-selector} equal to $w_0,w_1$, respectively.
\end{lemma}
\begin{proof}
First suppose that $(\alpha_T,\beta_T)\ne(\alpha_R,\beta_R)$.
The coefficients of the common coordinates $x,y$ in the two output values form the matrix
\[
 \begin{pmatrix}
 u+\alpha_T&v+\beta_T\\
 u+\alpha_R&v+\beta_R
 \end{pmatrix}.
\]
The determinant is a nonconstant affine function of $u,v$ because the two reports differ in at least one of their first two coordinates.
If the determinant is initially zero, an increment from $E_t$ in a coordinate with nonzero coefficient makes it nonzero.
This uses at most one control ballot.
With the first two coordinates fixed, solve the linear system for the common $x,y$ that give the desired output values.
Each required increment is the sum of at most $t$ elements of $E_t$, so at most $2t$ further control ballots suffice.

Otherwise, put $A=u+\alpha_T$, $B=v+\beta_T$, $\delta_x=\xi_R-\xi_T$, and
\[
 P_*(Z)=\eta_R-\eta_T+P_R(Z)-P_T(Z),\qquad \Delta=w_1-w_0\ne0.
\]
The difference between the reported and truthful values is
\[
 A\delta_x+B P_*(A).
\]
If $P_*\ne0$, its degree is at most $d$.
The zero increment and the first-coordinate increments of binary indices $1,\ldots,d$ give $d+1$ distinct possible values of $A$.
Thus, at most one control ballot makes $P_*(A)\ne0$.
Set $B=(\Delta-A\delta_x)/P_*(A)$ using at most $t$ ballots.
The difference is now $\Delta$, so $A,B$ cannot both be zero.
Adjust $x$ or $y$ using at most $t$ more ballots to make the truthful value $w_0$.
If $P_*=0$, both polynomial contributions must be zero.
Indeed, one nonzero polynomial cannot be canceled by the constant $\eta_R-\eta_T$.
If both polynomials were nonzero, their ballots would have zero vector contributions and distinct polynomials, again making $P_*\ne0$.
Thus, $\eta_T=\eta_R$.
Since the first two coordinates also agree but the contributions differ, we have $\delta_x\ne0$.
Set $A=\Delta/\delta_x\ne0$ and then adjust $x$, using at most $2t$ ballots.
\end{proof}

A single bound on $m$ ensures that enough control ballots are available for every committee size.
\begin{lemma}\label{lem:selector-parameters}
For every $m\ge4096$ and $2\le k<m$, the parameters above satisfy~\eqref{eq:binary-capacity}.
\end{lemma}
\begin{proof}
We have $M\ge m$ and $H\ge m$.
Hence, $r\ge12$ and $d\le m/12+1$.
For $t=2,3$, writing $r=tb+j$ with $0\le j<t$ gives
\[
 S_t=(t+j)2^b-t,\qquad H>2^{r-1}.
\]
For the two possible remainders when $t=2$ and the three when $t=3$, these inequalities give
\[
 S_2<3\sqrt H,\qquad S_3<4H^{1/3}.
\]

Suppose first that $k\ge3$, so $t=2$.
We claim that $H\le M^2/m$.
If $a=m-k$, then $H=\binom ma\le M\le M^2/m$.
Otherwise, $a=k-1\ge2$ and $a\le(m-1)/2$, so
\[
 \binom ma\ge\binom m2\ge\frac{m(m-a)}{a+1},\qquad
 H=\binom ma\frac{m-a}{a+1}\le\frac1m\binom ma^2\le\frac{M^2}{m}.
\]
Consequently,
\[
 4S_2+d<12\sqrt H+\frac m{12}+1
 \le\left(\frac3{16}+\frac1{12}+\frac1{4096}\right)M<M.
\]

For $k=2$, we have $M=m$, $H=\binom m2\le m^2/2$, and $t=3$.
Since $m^{1/3}\ge16$ and $2^{-1/3}<4/5$, we obtain
\[
 4S_3+d<16(m^2/2)^{1/3}+\frac m{12}+1
 <\left(\frac45+\frac1{12}+\frac1{4096}\right)m<m.
\]
This proves~\eqref{eq:binary-capacity} in both cases.
\end{proof}

\begin{proof}[Proof of Lemma~\ref{lem:seven-selector}]
Every output has a certificate satisfying the conditions in Proposition~\ref{prop:hash-certificates}, so it belongs to $\calR$.
Consider the committees $W,V$ in the lemma.
Their indices in the fixed committee order are distinct because their utilities differ.
The distinct reports $T,R$ also have different contributions.
The only collision in $g$ is between $\varnothing$ and $C$, and neither can be the true ballot of a strict-loss pair because its utility is constant over size-$k$ committees.

For $k\ge3$, put the two ballots from Lemma~\ref{lem:two-certificate-ballots} at two unknown positions.
The target certificates can then be found by~\eqref{eq:recover-certificate}, and changing the remaining positions within $\Sigma$ preserves them.
With $t=2$, Lemma~\ref{lem:four-selector} supplies at most five control ballots that make the target committees the proposed outputs.
Both committees pass the certificate test, giving the required strict loss with $2+5=7$ ballots.
For $k=2$, every certificate is a subset of size at most two, so no certificate ballots are needed.
The target committees remain in $\calR$ when the unknown empty ballots are replaced by control ballots.
With $t=3$, Lemma~\ref{lem:four-selector} uses at most seven such ballots to select the two targets.
In either case, make any unused positions empty.

It remains to check anonymity and polynomial time.
All contributions depend only on ballot types, and the state is their sum.
The certificate test and $W^{\mathrm{fb}}$ also depend only on the ballot multiset, so the procedure is anonymous.
The field representation has $r=O(m)$ bits per element, and every polynomial has degree $d\le m$.
The values $H,M,S_t$ have $O(m)$ bits and can be computed in polynomial time.
No list of the exponentially many ballots or committees is needed.
For $0\le i_1<\cdots<i_j<m$, the colexicographic index of $\{c_{i_1},\ldots,c_{i_j}\}$, starting from zero, is $\sum_{\ell=1}^j\binom{i_\ell}{\ell}$.
Adding the counts of smaller nonempty complements gives a ballot's index in the ordered family $\calL\setminus\{C\}$.
These formulas and binary search also recover a set from its index in polynomial time.
Within each block of $E_t$, field indices are consecutive positive integers multiplied by the power of two at the block's first coordinate.
We can therefore count the control entries below a given field index directly, also checking the at most $d$ additional first-coordinate entries.
These counts and binary search recover a control vector from its index.
Thus, a ballot's contribution can be computed directly; for a ballot with a polynomial contribution, its coefficients follow from its binary index.
Finally, the certificate test uses $O(m^2+nm^4)$ candidate sets and a polynomial number of maximum-matching computations in graphs with at most $n+mq$ vertices and $nmq$ edges.
The fallback committee is also computable in polynomial time by Proposition~\ref{prop:hash-certificates}.
Hence, the entire procedure runs in polynomial time.
\end{proof}

\end{document}